\documentclass[11pt,a4paper]{article}

\usepackage{natbib}
\usepackage{bm}
\usepackage[left = 3cm, right = 3cm, top = 2.5cm, bottom = 3cm]{geometry}
\usepackage{amsmath}
\usepackage{tabularx}
\usepackage{amssymb}
\usepackage{amsthm}
\usepackage{pdfsync}
\usepackage{booktabs}
\usepackage{multirow}
\usepackage{graphicx}
\usepackage{mathdots}
\usepackage{url}

\newcommand*{\bb}{\boldsymbol}
\newcommand*{\hderstraight}[2]{\mathrm{d}{#1}/\mathrm{d}{#2} }

\newcommand{\refer}[1]{(\ref{#1})}

\newcommand*{\hessian}[2]{\mathcal{D}^2 \left({#1};{#2}\right)}

\newtheoremstyle{example}
{3pt}		
{3pt}		
{}	        
{0\parindent}	
{\bf}
{:}		
{.5em}		
{}		
\newtheoremstyle{theorem}
{3pt}		
{3pt}		
{\em}	        
{0\parindent}	
{\bf}
{:}		
{.5em}		
{}		
\theoremstyle{example}
\newtheorem{example}{Example}[section]
\theoremstyle{theorem}
\newtheorem{theorem}{Theorem}[section]

\def\trace{\mathop{\rm tr}}

\title{Improved estimation in cumulative link models}

\author{Ioannis Kosmidis\\
Department of Statistical Science, University College London \\
London, WC1E 6BT, UK \\
\texttt{i.kosmidis@ucl.ac.uk}}

\def\log{\mathop{\rm log}}
\def\trace{\mathop{\rm tr}}

\def\diag{\mathop{\rm diag}}

\begin{document}
\maketitle
\begin{abstract}
  For the estimation of cumulative link models for ordinal data, the
  bias-reducing adjusted score equations in \citet{firth:93} are
  obtained, whose solution ensures an estimator with smaller
  asymptotic bias than the maximum likelihood estimator. Their form
  suggests a parameter-dependent adjustment of the multinomial counts,
  which, in turn suggests the solution of the adjusted score equations
  through iterated maximum likelihood fits on adjusted counts, greatly
  facilitating implementation. Like the maximum likelihood estimator,
  the reduced-bias estimator is found to respect the invariance
  properties that make cumulative link models a good choice for the
  analysis of categorical data. Its additional finiteness and optimal
  frequentist properties, along with the adequate behaviour of
  related asymptotic inferential procedures make the reduced-bias
  estimator attractive as a default choice for practical
  applications. Furthermore, the proposed estimator enjoys certain
  shrinkage properties that are defensible from an experimental point
  of view relating to the nature of ordinal data. 
\end{abstract}


\section{Introduction}
In many models with categorical responses the maximum likelihood
estimates can be on the boundary of the parameter space with positive
probability. For example, \citet{albert:84} derive the conditions that
describe when the maximum likelihood estimates are on the boundary in
multinomial logistic regression models. While there is no ambiguity in
reporting an estimate on the boundary of the parameter space, as is
for example an infinite estimate for the parameters of a logistic
regression model, estimates on the boundary can (i) cause numerical
instabilities to fitting procedures, (ii) lead to misleading output
when estimation is based on iterative procedures with a stopping
criterion, and more importantly, (iii) cause havoc to asymptotic
inferential procedures, and especially to the ones that depend on
estimates of the standard error of the estimators (for example, Wald
tests and related confidence intervals).

The maximum likelihood estimator in cumulative link models for ordinal
data \citep{mccullagh:80} also has a positive probability of being on
the boundary.

\begin{example}
\label{wineExample}
As a demonstration consider the example in \citet[][\S
7]{christensen:12}. The data set in Table~\ref{wineData} comes from
\citet{randall:89} and concerns a factorial experiment for
investigating factors that affect the bitterness of white wine. There
are two factors in the experiment, temperature at the time of crushing
the grapes (with two levels, ``cold'' and ``warm'') and contact of the
juice with the skin (with two levels ``Yes'' and ``No''). For each
combination of factors two bottles were rated on their bitterness by a
panel of 9 judges. The responses of the judges on the bitterness of
the wine were taken on a continuous scale in the interval from 0
(``None'') to 100 (``Intense'') and then they were grouped
correspondingly into $5$ ordered categories, $1$, $2$, $3$, $4$, $5$.

\begin{table}[t!]
  \caption{The top table contains the wine tasting data
    \citep{randall:89} (top). The bottom table contains the
    maximum likelihood estimates for the parameters of
    model~(\ref{partialProp}), the corresponding estimated standard
    errors (in parenthesis) and the values of the $Z$ statistic
    (bottom) for the hypothesis that the corresponding parameter is
    zero.}
  \begin{small}
  \begin{center}
    \begin{tabular}{ccccccc}
      \midrule\midrule
      \multirow{2}{*}{Temperature} & \multirow{2}{*}{Contact} &
      \multicolumn{5}{c}{Bitterness scale} \\
      & & 1 & 2 & 3 & 4 & 5 \\ \midrule
      Cold & No  & 4 & 9 & 5 & 0 & 0 \\
      Cold & Yes & 1 & 7 & 8 & 2 & 0 \\
      Warm & No  & 0 & 5 & 8 & 3 & 2 \\
      Warm & Yes & 0 & 1 & 5 & 7 & 5 \\ \midrule
    \end{tabular}
    \\
    \begin{tabular}{crrr}
      \midrule
      Parameter & \multicolumn{2}{c}{ML estimates} & $Z$-statistic \\ \midrule
      $\alpha_1$ & -1.27 & (0.51) & -2.46 \\
      $\alpha_2$ &  1.10 & (0.44) & 2.52 \\
      $\alpha_3$ &  3.77 & (0.80) & 4.68 \\
      $\alpha_4$ & 28.90 & (193125.63) & 0.00 \\
      $\beta_1$ & 25.10 & (112072.69) & 0.00 \\
      $\beta_2$ &  2.15 & (0.59) & 3.65 \\
      $\beta_3$ &  2.87 & (0.82) & 3.52 \\
      $\beta_4$ & 26.55 & (193125.63) & 0.00 \\
      $\theta$ & 1.47 & (0.47) & 3.13 \\
      \midrule\midrule
\end{tabular}
\end{center}
\end{small}
\label{wineData}
\end{table}

Consider the partial proportional odds model \citep{peterson:90} with
\begin{equation}
  \label{partialProp}
  \log{\frac{\gamma_{rs}}{1 - \gamma_{rs}}} = \alpha_s - \beta_s w_r -
  \theta z_r \quad (r = 1, \ldots, 4; s = 1, \ldots, 4)\, ,
\end{equation}
where $w_r$ and $z_r$ are dummy variables representing the factors
temperature and contact, respectively, $\alpha_1, \ldots, \alpha_4,
\beta, \theta$ are model parameters and $\gamma_{rs}$ is the
cumulative probability for the $s$th category at the $r$th combination
of levels for temperature and contact. The \texttt{clm} function of
the R package \texttt{ordinal} \citep{ordinal:12} is used to maximize
the multinomial likelihood that corresponds to
model~(\ref{partialProp}). The \texttt{clm} function finds the maximum
likelihood estimates up to a specified accuracy, by using a
Newton-Raphson iteration for finding the roots of the log-likelihood
derivatives. For the current example, the stopping criterion is set to
that the log-likelihood derivatives are less than $10^{-10}$ in
absolute value.  The maximum likelihood estimates, estimated standard
errors and the corresponding values for the $Z$ statistics for the
hypothesis that the respective parameter is zero, are extracted from
the software output and shown in Table~\ref{wineData}. At those values
for the maximum likelihood estimator the maximum absolute
log-likelihood derivative is less than $10^{-10}$ and the software
correctly reports convergence. Nevertheless, an immediate observation
is that the absolute value of the estimates and estimated standard
errors for the parameters $\alpha_4$, $\beta_1$ and $\beta_4$ is very
large. Actually, these would diverge to infinity as the stopping
criteria of the iterative fitting procedure used become stricter and
the number of allowed iterations increases.

For model~(\ref{partialProp}) interest usually is on testing
departures from the assumption of proportional odds via the hypothesis
$\beta_1 = \beta_2 = \beta_3 = \beta_4$. Using a Wald-type statistic
would be adventurous here because such a statistic explicitly depends
on the estimates of $\beta_1$, $\beta_2$, $\beta_3$ and $\beta_4$. Of
course, given that the likelihood is close to its maximal value at the
estimates in Table~\ref{wineData}, a likelihood ratio test can be used
instead; the likelihood ratio test for the particular example has been
carried out in \citet[][\S 7]{christensen:12}.

Furthermore, the current example demonstrates some of the potential
dangers involved in the application of cumulative link models in
general; the behaviour of the individual $Z$ statistics --- being
essentially $0$ for the parameters $\beta_1$ and $\beta_4$ in this
example --- is quite typical of what happens when estimates diverge to
infinity. The values of the $Z$ statistics converge to zero because
the estimated standard errors diverge much faster than the estimates,
irrespective of whether or not there is evidence against the
individual hypotheses. This behaviour is also true for individual
hypotheses at values other than zero and can lead to invalid
conclusions if the output is interpreted naively. More importantly,
the presence of three infinite standard errors in a non-orthogonal
\citep[in the sense of][]{cox:87} setting like the current may affect
the estimates of the standard errors for other parameters in ways that
are hard to predict. \qed
\end{example}

An apparent solution to the issues mentioned in
Example~\ref{wineExample} is to use a different estimator that has
probability zero of resulting in estimates on the boundary of the
parameter space. For example, for the estimation of the common
difference in cumulative logits from ordinal data arranged in a
$2\times k$ contingency table with fixed row totals,
\citet{mccullagh:80} describes the generalized empirical logistic
transform. The generalized empirical logistic transform has smaller
asymptotic bias than the maximum likelihood estimator and is also
guaranteed to give finite estimates of the difference in cumulative
logits because it adjusts all cumulative counts by $1/2$. However, the
applicability of this estimator is limited to the analysis of $2\times
k$ tables, and particularly in estimating differences in cumulative
logits, with no obvious extension to more general cumulative link
models, such as the one in Example~\ref{wineExample}.

A family of estimators that can be used for arbitrary cumulative link
models and are guaranteed to be finite are ridge estimators. As one of
the referees highlighted, if one extends the work in
\citet{lecessie:92} for logistic regression to cumulative link models,
then the shrinkage properties of the ridge estimator can guarantee its
finiteness. Nevertheless, ridge estimators have a series of
shortcomings. Firstly, in contrast to the maximum likelihood
estimator, the ridge estimators are not generally equivariant under
general linear transformations of the parameters for cumulative link
models. A reparameterization of the model by re-scaling the parameters
or taking contrasts of those --- which are often interesting
transformations in cumulative link models --- and a corresponding
transformation of the ridge estimator will not generally result to the
estimator that the same ridge penalty would produce for the
reparameterized model, unless the penalty is also appropriately
adjusted. For example, if one wishes to test the hypothesis in
Example~\ref{wineExample} using a Wald test, then an appropriate ridge
estimator would be one that penalizes the size of the contrasts of
$\beta_1$, $\beta_2$, $\beta_3$ and $\beta_4$ instead of the size of
those parameters themselves. Secondly, the choice of the tuning
parameter is usually performed through the use of an optimality
criterion and cross-validation. Hence, the properties of the resultant
estimators are sensitive to the choice of the criterion. For example,
criteria like mean-squared error of predictions, classification error,
and log-likelihood that have been discussed in \citet{lecessie:92}
will each produce different results, as is also shown in the latter
study. Furthermore, the resultant ridge estimator is sensitive to the
type of cross-validation used. For example, $k$-fold cross-validation
will produce different results for different choices of $k$. Lastly,
standard asymptotic inferential procedures for performing hypothesis
tests and constructing confidence intervals cannot be used by simply
replacing the maximum likelihood estimator with the ridge estimator in
the associated pivots. For the above reasons, ridge estimators can
only offer an ad-hoc solution to the problem.

Several simulation studies on well-used models for discrete responses
have demonstrated that bias reduction via the adjustment of the
log-likelihood derivatives \citep{firth:93} offers a solution to the
problems relating to boundary estimates; see, for example,
\citet{mehrabi:95} for the estimation of simple complementary log-log
models, \citet{heinze:02} and \citet{bull:02,kosmidis:11} for
binomial and multinomial logistic regression, respectively, and
\citet{kosmidis:09a} for binomial-response generalized linear models.

In the current paper the aforementioned adjustment is derived and
evaluated for the estimation of cumulative link models for ordinal
responses. It is shown that reduction of bias is equivalent to a
parameter-dependent additive adjustment of the multinomial counts and
that such adjustment generalizes well-known constant adjustments in
cases like the estimation of cumulative logits. Then, the reduced-bias
estimates can be obtained through iterative maximum likelihood fits to
the adjusted counts. The form of the parameter-dependent adjustment is
also used to show that, like the maximum likelihood estimator, the
reduced-bias estimator is invariant to the level of sample aggregation
present in the data.

Furthermore, it is shown that the reduced-bias estimator respects the
invariance properties that make cumulative link models an attractive
choice for the analysis of ordinal data. The finiteness and shrinkage
properties of the proposed estimator are illustrated via detailed
complete enumeration and an extensive simulation exercise. In
particular, the reduced-bias estimator is found to be always finite,
and also the reduction of bias in cumulative link models results in
the shrinkage of the multinomial model towards a smaller binomial
model for the end-categories. A thorough discussion on the desirable
frequentist properties of the estimator is provided along with an
investigation of the performance of associated inferential procedures.

The finiteness of the reduced-bias estimator, its optimal frequentist
properties and the adequate performance of the associated inferential
procedures lead to its proposal for routine use in fitting cumulative
link models.

The exposition of the methodology is accompanied by a parallel
discussion of the corresponding implications in the application of the
models through examples with artificial and real data.

\section{Cumulative link models}
\label{cumlink}
Suppose observations of $n$ $k$-vectors of counts $y_1, \ldots, y_n$,
on mutually independent multinomial random vectors $Y_1, \ldots, Y_n$,
where ${\bb Y}_r = (Y_{r1}, \ldots, Y_{rk})^T$ and the $k$ multinomial
categories are ordered. The multinomial totals for $Y_r$ are $m_r =
\sum_{s = 1}^k y_{rs}$ and the probability for the $s$th category of
the $r$th multinomial vector is $\pi_{rk}$, with $\sum_{s = 1}^k
\pi_{rs} = 1$ $(r = 1, \ldots, n)$. The cumulative link model links
the cumulative probability $\gamma_{rs} = \pi_{r1} + \ldots +
\pi_{rs}$ to a $p$-vector of covariates ${\bb x}_r$ via the relationship
\begin{equation}
\label{model}
\gamma_{rs} = G\left(\alpha_s - \sum_{t = 1}^p\beta_t x_{rt}\right)
\quad (s = 1, \ldots, q; \,r = 1, \ldots, n) \, ,
\end{equation}
where $q = k -1$ denotes the number of the non-redundant components of
$y_r$, and where ${\bb\delta} = (\alpha_1, \ldots, \alpha_q, \beta_1,
\ldots, \beta_p)^T$ is a $(p + q)$-vector of real-valued model
parameters, with $\alpha_1 < \ldots < \alpha_q$. The function $G(.)$
is a monotone increasing function mapping $(-\infty, +\infty)$ to $(0,
1)$, usually chosen to be a known distribution function (like, for
example, the logistic, extreme value or standard normal distribution
function). Then, $\alpha_1, \ldots, \alpha_q$ can be considered as
cutpoints on the latent scale implied by $G$.

Special important cases of cumulative link models are the
proportional-odds model with $G(\eta) = \exp(\eta)/\{1+\exp(\eta)\}$,
and the proportional hazards model in discrete time with $G(\eta) = 1
- \exp\left\{-\exp(\eta)\right\}$ \citep[see,][for the introduction of
and a thorough discussion on cumulative link models]{mccullagh:80}.

The cumulative link model can be written in the usual multivariate
generalized linear models form by writing the relationship that links
the cumulative probability $\gamma_{rs}$ to ${\bb\delta}$ as
\begin{equation}
\label{glmForm}
G^{-1}(\gamma_{rs}) = \eta_{rs} =\sum_{t = 1}^{p + q} \delta_t z_{rst}
\quad (s = 1, \ldots, q;\, r = 1, \ldots, n) \, ,
\end{equation}
where $z_{rst}$ is the $(s, t)$th component of the $q \times (p + q)$
matrix
\[
Z_r = \left[
\begin{array}{ccccc}
1 & 0 & \ldots & 0 & -{\bb x}_{r}^T \\
0 & 1 & \ldots & 0 & -{\bb x}_{r}^T \\
\vdots & \vdots & \ddots & \vdots & \vdots \\
0 & 0 & \ldots & 1 & -{\bb x}_{r}^T \\
\end{array}
\right] \quad (r = 1, \ldots, n) \,.
\]
In order to be able to identify ${\bb\delta}$, the matrix $Z$
with row blocks $Z_1, \ldots, Z_n$ is assumed to be of full rank.

Direct differentiation of the multinomial log-likelihood
$l({\bb\delta})$ gives that the $t$th component of the vector
of score functions has the form
\begin{equation}
\label{score}
U_t({\bb\delta}) = \sum_{r = 1}^n\sum_{s = 1}^q g_{rs}({\bb\delta})
\left(\frac{y_{rs}}{\pi_{rs}({\bb\delta})} -
  \frac{y_{rs+1}}{\pi_{rs+1}({\bb\delta})} \right)z_{rst} \quad (t = 1,
\ldots, p + q)\,,
\end{equation}
where $g_{rs}({\bb\delta}) = g(\eta_{rs})$, with $g(\eta) =
\hderstraight{G(\eta)}{\eta}$. If $g(.)$ is log-concave then
$U_t(\hat{\bb\delta}) = 0$ $(t = 1, \ldots, p + q)$ has unique
solution the maximum likelihood estimate $\hat{\bb\delta}$
\citep[see,][where it is shown that the log-concavity of g(.) implies
the concavity of $l({\bb\delta})$]{pratt:81}.

All generalized linear models for binomial responses that include an
intercept parameter in the linear predictor are special cases of
model~(\ref{model}).

\section{Maximum likelihood estimates on the boundary}
\label{infinite}
The maximum likelihood estimates of the parameters of the cumulative
link model can be on the boundary of the parameter space with positive
probability. Under the log-concavity of $g(.)$, \citet{haberman:80}
gives conditions that guarantee that the maximum likelihood estimates
are not on the boundary (``exist'' in an alternative
terminology). Boundary estimates for these models are estimates of the
regression parameters $\beta_1, \ldots, \beta_p$ with an infinite
value, and/or estimates of the cutpoints $-\infty = \alpha_0 <
\alpha_1 < \ldots < \alpha_q < \alpha_k = \infty$ for which at least a
pair of consecutive cutpoints have equal estimated value.

As far as the regression parameters ${\bb\beta}$ are concerned,
\citet[][\S~3.4.5]{agresti:10a} gives an accessible account on what
data settings result in infinite estimates for the regression
parameters, how a fitted model with such estimates can be used for
inference and how such problems can be identified from the output of
standard statistical software.

As far as boundary values of the cutpoints ${\bb\alpha}$ are concerned,
\citet{pratt:81} showed that with a log-concave $g(.)$, the cutpoints
$\alpha_{s-1}$ and $\alpha_{s}$ have equal estimates if and only if
the observed counts for the $s$th category are zero $(s = 1, \ldots,
k)$ for all $r \in \{1, \ldots, n\}$. If the first or the last
category have zero counts then the respective estimates for $\alpha_1$
and $\alpha_q$ will be $-\infty$ and $+\infty$, respectively, and if
this happens for category $s$ for some $s\in\{2, \ldots, q\}$, then the
estimates for $\alpha_{s-1}$ and $\alpha_s$ will have the same finite
value.

\section{Bias correction and bias reduction}
\subsection{Adjusted score functions and first-order bias}
Denote by $b({\bb\delta})$ the first term in the asymptotic expansion
of the bias of the maximum likelihood estimator in decreasing orders
of information, usually sample-size. Call $b({\bb\delta})$ the
first-order bias term, and let $F({\bb\delta})$ denote the expected
information matrix for $\delta$. \citet{firth:93} showed that, if
$A({\bb\delta}) = -F({\bb\delta})b({\bb\delta})$ then the solution of the
adjusted score equations
\begin{equation}
\label{adjusted.score}
U_t^*({\bb\delta}) = U_t({\bb\delta}) + A_t({\bb\delta}) = 0 \quad (t
= 1, \ldots, q + p) \, ,
\end{equation}
results in an estimator that is free from the first-order term in the
asymptotic expansion of its bias.

\subsection{Reduced-bias estimator}
\label{adjscores}
\citet{kosmidis:09} exploited the structure of the bias-reducing
adjusted score functions in \refer{adjusted.score} in the case of
exponential family non-linear models. Using \citet[][expression
(9)]{kosmidis:09} for the adjusted score functions in the case of
multivariate generalized linear models, and temporarily omitting the
argument ${\bb\delta}$ from the quantities that depend on it, the
adjustment functions $A_t$ in \refer{adjusted.score} have the form
\begin{equation}
\label{adjustment}
A_t = \frac{1}{2}\sum_{r = 1}^nm_r\sum_{s = 1}^q\trace\left[V_r
  \left\{\left(D_r\Sigma_r^{-1}\right)_s \otimes 1_q
  \right\}\hessian{{\bb\pi}_r}{{\bb\eta}_r}\right]z_{rst} \quad (t = 1, \ldots,
q + p) \, ,
\end{equation}
where $V_r = Z_rF^{-1}Z_r^T$ is the asymptotic variance-covariance
matrix of the estimator for the vector of predictor functions
${\bb\eta}_r = (\eta_{r1}, \ldots, \eta_{rq})^T$ and ${\bb\pi}_r =
(\pi_{r1}, \ldots, \pi_{rq})^T$. Furthermore,
$\hessian{{\bb\pi}_r}{{\bb\eta}_r}$ is the $q^2 \times q$ matrix with
$s$th block the hessian of ${\bb\pi}_{rs}$ with respect to ${\bb
  \eta}_r$ $(s = 1, \ldots, q)$, $1_q$ is the $q \times q$ identity
matrix and $D_r^T$ is the $q \times q$ Jacobian of $m_r{\bb\pi}_r$
with respect to ${\bb\eta}_r$. A straightforward calculation shows
that
\[
D_r^T = m_r\left[
\begin{array}{ccccc}
g_{r1} & 0 & \ldots & 0 & 0 \\
-g_{r1} & g_{r2} & \ldots & 0 & 0 \\
0 & -g_{r2} & \ddots & \vdots & \vdots \\
\vdots & \vdots & \ddots & g_{rq-1} & 0 \\
0 & 0 & \ldots & -g_{rq-1} & g_{rq}
\end{array}
\right]\quad (r = 1, \ldots, n)\, .
\]
The matrix $\Sigma_r$ is the incomplete $q\times q$
variance-covariance matrix of the multinomial vector ${\bb Y}_r$ with $(s,
u)$th component
\[
\sigma_{rsu} = \left\{
  \begin{array}{ll}
    m_r\pi_{rs}(1 - \pi_{rs}) \, , & \quad s = u \\
    -m_r\pi_{rs}\pi_{ru} \, , & \quad s\ne u \\
  \end{array}\right. \quad (s, u = 1, \ldots, q;\, r=1,\ldots,n) \, .
\]

Substituting in \refer{adjustment}, some tedious calculation gives
that the adjustment functions $A_t$ have the form
\begin{equation}
\label{adjustment2}
A_t = \sum_{r = 1}^n\sum_{s = 1}^q g_{rs}
\left(\frac{c_{rs}- c_{rs-1}}{\pi_{rs}} -
  \frac{c_{rs+1} - c_{rs}}{\pi_{rs+1}} \right)z_{rst} \quad (t = 1,
\ldots, q+p)\,,
\end{equation}
where
\begin{equation}
\label{add.adj}
c_{r0} = c_{rk} = 0 \quad \text{and} \quad
c_{rs} = \frac{1}{2}m_rg'_{rs}v_{rss} \quad (s = 1, \ldots, q)\, ,
\end{equation}
with $g'_{rs} = g'(\eta_{rs})$, and $g'(\eta) =
\hderstraight{^2G(\eta)}{\eta^2}$, The quantity $v_{rss}$ is the $s$th
diagonal component of the matrix $V_r$ $(s = 1, \ldots, q;\, r = 1,
\ldots, n)$.

Substituting \refer{score} and \refer{adjustment2} in
\refer{adjusted.score} gives that the $t$th component of the
bias-reducing adjusted score vector $(t = 1, \ldots, q+p)$ has the
form
\begin{equation}
\label{adjustedScore2}
U_t^*({\bb\delta}) = \sum_{r = 1}^n\sum_{s = 1}^q g_{rs}({\bb\delta})
\left\{\frac{y_{rs} + c_{rs}({\bb\delta}) - c_{rs-1}({\bb\delta})}{\pi_{rs}({\bb\delta})} -
  \frac{y_{rs+1} + c_{rs+1}({\bb\delta}) -
    c_{rs}({\bb\delta})}{\pi_{rs+1}({\bb\delta})} \right\}z_{rst} \, .
\end{equation}
The reduced-bias estimates $\tilde{\bb\delta}_{RB}$ are such that
$U_t^*(\tilde{\bb\delta}_{RB}) = 0$ for every $t \in \{1 = 1, \ldots,
q+p\}$. \citet[][Chapter 6]{kosmidis:07} shows that if the maximum
likelihood is consistent, then the reduced-bias estimator is also
consistent. Furthermore, $\tilde{\bb\delta}_{RB}$ has the same
asymptotic distribution as $\hat{\bb\delta}$, namely a multivariate
Normal distribution with mean ${\bb\delta}$ and variance-covariance
matrix $F^{-1}({\bb\delta})$. Hence, estimated standard errors for
$\tilde{\bb\delta}_{RB}$ can be obtained as usual by using the square roots
of the diagonal elements of the inverse of the Fisher information at
$\tilde{\bb\delta}_{RB}$. All inferential procedures that rely in the
asymptotic Normality of the estimator can directly be adapted to the
reduced-bias estimator.

\subsection{Bias-corrected estimator}
Expression \refer{adjustment2} can also be used to evaluate the
first-order bias term as ${\bb b}({\bb\delta}) =
-F^{-1}({\bb\delta}){\bb A}({\bb\delta})$, where $F({\bb\delta}) =
\sum_{r =
  1}^nZ_r^TD_r({\bb\delta})\Sigma_r^{-1}({\bb\delta})D_r^T({\bb\delta})Z_r$. If
$\hat{\bb\delta}$ is the maximum likelihood estimator then
\begin{equation}
\label{biasCorrection}
\tilde{\bb\delta}_{BC} = \hat{\bb\delta} - {\bb b}(\hat{\bb\delta})
\end{equation}
is the bias-corrected estimator which has been studied in
\citet{cordeiro:91} for univariate generalized linear models. The
estimator $\tilde{\bb\delta}_{BC}$ can be shown to have no first-order term
in the expansion of its bias \citep[see,][for analytic derivation of
this result]{efron:75}.

\subsection{Models for binomial responses}
\label{binomModels}
For $k = 2$, $Y_{r1}$ has a Binomial distribution with index $m$ and
probability $\pi_{r1}$, and $Y_{r2} = m_r - Y_{r1}$. Then
model~(\ref{model}) reduces to the univariate generalized linear model
\[
G(\pi_r) = \alpha - \sum_{t = 1}^p\beta_t x_{rt} \quad (r = 1, \ldots, n) \, .
\]
From (\ref{adjustedScore2}), the adjusted score functions take the
form
\[
U_t^*({\bb\delta}) = \sum_{r = 1}^n g_{r1}({\bb\delta}) \left\{\frac{y_{r1} +
    c_{r1}({\bb\delta})}{\pi_{r1}({\bb\delta})} - \frac{m_r - y_{r1} -
    c_{r1}({\bb\delta})}{1 - \pi_{r1}({\bb\delta})} \right\}z_{r1t} \quad (t = 1,
\ldots, p + 1) \, .
\]
Omitting the category index for notational simplicity, a re-expression
of the above equality gives that the adjusted score functions for
binomial generalized linear models have the form
\begin{equation}
\label{adj.score.binom}
U_t^*({\bb\delta}) = \sum_{r=1}^n \frac{g_r}{\pi_r(1-\pi_r)}\left(y_r +
  \frac{g_r'}{2w_r}h_r - m_r \pi_r\right)z_{rt} \quad (t = 1,
\ldots, p + 1)\, ,
\end{equation}
where $w_r = m_r g_r^2/\{\pi_r(1-\pi_r)\}$ are the working weights and
$h_r$ is the $r$th diagonal component of the ``hat'' matrix $H =
ZF^{-1}Z^TW$, with $W = \diag\{w_1, \ldots, w_n\}$ and
\[
Z = \left[
  \begin{array}{ccc}
    1 & -{\bb x}_1^T \\
    1 & -{\bb x}_2^T \\
    \vdots & \vdots \\
    1 & -{\bb x}_n^T
  \end{array}
\right]\,.
\]
The above expression agrees with the results in \citet[][\S
4.3]{kosmidis:09}, where it is shown that for generalized linear
models reduction of bias via adjusted score functions is equivalent to
replacing the actual count $y_r$ with the parameter-dependent
adjusted count $y_r + g_r'h_r/(2w_r)$ $(r = 1, \ldots, n)$.

\section{Implementation}
\label{implementation}
\subsection{Maximum likelihood fits on iteratively adjusted counts}
\label{adjustedCounts}
When expression \refer{adjustedScore2} is compared to expression
\refer{score}, it is directly apparent that bias-reduction is
equivalent to the additive adjustment of the multinomial count
$y_{rs}$ by the quantity $c_{rs}({\bb\delta}) - c_{rs-1}({\bb\delta})$ $(s = 1,
\ldots, k; r = 1, \ldots, n)$. Noting that these quantities depend on
the model parameters in general, this interpretation of bias-reduction
can be exploited to set-up an iterative scheme with a stationary point
at the reduced-bias estimates: at each step, i) evaluate $y_{rs} +
c_{rs}({\bb\delta}) - c_{rs-1}({\bb\delta})$ at the current value of ${\bb\delta}$
$(s = 1, \ldots, q; r = 1, \ldots, n)$, and ii) fit the original model
to the adjusted counts using some standard maximum likelihood routine.

However, $c_{rs}({\bb\delta}) - c_{rs-1}({\bb\delta})$ can take negative values
which in turn may result in fitting the model on negative counts in
step ii) above. In principle this is possible but then the
log-concavity of $g(.)$ does not necessarily imply concavity of the
log-likelihood function and problems may arise when performing the
maximization in ii) \citep[see, for example,][where the transition
from the log-concavity of $g(.)$ to the concavity of the likelihood
requires that the latter is a weighted sum with non-negative
weights]{pratt:81}.  That is the reason why many published maximum
likelihood fitting routines will complain if supplied with negative
counts.

The issue can be remedied through a simple calculation. Temporarily
omitting the index $r$, let $a_s = c_{s}- c_{s-1}$ $(s = 1, \ldots,
k)$. Then the kernel $(y_{s} + a_s)/\pi_{s} - (y_{s} +
a_{s+1})/\pi_{s+1}$ in \refer{adjustedScore2} can be re-expressed as
\begin{align*}
\frac{y_{s} + a_sI(a_{s}>0)
  -\pi_{s}a_{s+1}I(a_{s+1}\le 0)/\pi_{s+1}}{\pi_{s}} -
\frac{y_{s+1} + a_{s+1}I(a_{s+1}>0)
  -\pi_{s+1}a_{s}I(a_{s} \le 0)/\pi_{s}}{\pi_{s+1}} \, ,
\end{align*}
where $I(E) = 1$ if $E$ holds and $0$ otherwise. Note that
\[
a_s({\bb\delta})I(a_{s}({\bb\delta})>0)
-\pi_{s}({\bb\delta})a_{s+1}({\bb\delta})I(a_{s+1}({\bb\delta})<0)/\pi_{s+1}({\bb\delta})
\ge 0 \, ,
\]
uniformly in ${\bb\delta}$. Hence, if step i) in the above procedure
adjusts $y_{rs}$ by $a_{rs}I(a_{rs}>0)
-\pi_{rs}a_{rs+1}I(a_{rs+1}<0)/\pi_{rs+1}$ evaluated at the current
value of ${\bb\delta}$, then the possibility of issues relating to negative
adjusted counts in step ii) is eliminated, and the resultant iterative
procedure still has a stationary point at the reduced-bias estimates.

\subsection{Iterative bias correction}
Another way to obtain the reduced-bias estimates is via the iterative
bias-correction procedure of \citet{kosmidis:10}; if the current value
of the estimates is ${\bb\delta}^{(i)}$ then the next candidate value is
calculated as
\begin{equation}
\label{quasifisher}
{\bb\delta}^{(i+1)} = \hat{\bb\delta}^{(i+1)} - b\left({\bb\delta}^{(i)}\right) \quad
(i = 0, 1, \ldots)\, ,
\end{equation}
where $\hat{\bb\delta}^{(i+1)} = \hat{\bb\delta}^{(i)} +
F^{-1}\left({\bb\delta}^{(i)}\right)U\left({\bb\delta}^{(i)}\right)$, that is
$\hat{\bb\delta}^{(i+1)}$ is the next candidate value for the maximum
likelihood estimator obtained through a single Fisher scoring step,
starting from ${\bb\delta}^{(i)}$.

Iteration (\ref{quasifisher}) generally requires more effort in
implementation than the iteration described in the
Subsection~\ref{adjustedCounts}. Nevertheless, if the starting value
${\bb\delta}^{(0)}$ is chosen to be the maximum likelihood estimates then
the first step of the procedure in (\ref{quasifisher}) will result in
the bias-corrected estimates defined in (\ref{biasCorrection}).

\section{Additive adjustment of the multinomial counts}
\subsection{Estimation of cumulative logits}
\label{cumlogits}
For the estimation of the cumulative logits $\alpha_s =
\log\{\gamma_s/(1-\gamma_s)\}$ $(s = 1, \ldots, q)$ from a single
multinomial observation $y_{1}, \ldots, y_{k}$ the maximum likelihood
estimator of $\alpha_s$ $(s = 1, \ldots, q)$ is $\hat\alpha_s =
\log\{R_{s}/(m - R_{s})\}$, where $R_{s} = \sum_{j = 1}^sY_{s}$ is the
$s$th cumulative count. The Fisher information for $\alpha_1, \ldots,
\alpha_q$ is the matrix of quadratic weights $W =
D\Sigma^{-1}D^T$. The matrix $W$ is symmetric and tri-diagonal with
non-zero components
\begin{align*}
  W_{ss} & = m\gamma_s^2(1 - \gamma_s)^2\left(\frac{1}{\gamma_s -
      \gamma_{s-1}} + \frac{1}{\gamma_{s+1} - \gamma_{s}}\right)\quad (s = 1, \ldots, q) \\
  W_{s-1,s} & = -m\frac{\gamma_{s-1}(1 - \gamma_{s-1})\gamma_s(1 -
    \gamma_s)}{\gamma_s - \gamma_{s-1}} \quad (s = 2, \ldots, q) \, ,
\end{align*}
with $\gamma_0 = 0$ and $\gamma_k = 1$. By use of the recursion
formulae in \citet{usmani:94} for the inversion of a tri-diagonal
matrix, the $s$th diagonal component of $F^{-1} = W^{-1}$ is
$1/(m\gamma_s(1-\gamma_s))$. Hence, using (\ref{add.adj}) and noting
that $g_s = \gamma_s(1-\gamma_{s})(1-2\gamma_{s})$ for the logistic
link, $c_s = \frac{1}{2} - \gamma_s$ $(s = 1, \ldots,
q)$. Substituting in (\ref{adjustedScore2}) yields that reduction of
bias is equivalent to adding $1/2$ to the counts for the first and the
last category and leaving the rest of the counts unchanged.

The above adjustment scheme reproduces the empirical logistic
transforms $\tilde\alpha_s = \log\{(R_s + 1/2)/(m - R_s + 1/2)\}$,
which are always finite and have smaller asymptotic bias than
$\hat\alpha_s$ \citep[see][\S 2.1.6, under the fact that the marginal
distribution of $R_s$ given $R_k = m$ is Binomial with index $m$ and
probability $\gamma_s$ for any $s \in \{1, \ldots, q\}$]{cox:89}.

\subsection{A note of caution for constant adjustments in general
  settings}
Since the works of \citet{haldane:55} and \citet{anscombe:56}
concerning the additive modification of the binomial count by $1/2$
for reducing the bias and guaranteeing finiteness in the problem of
log-odds estimation, the addition of small constants to counts when
the data are sparse has become a standard practice for avoiding
estimates on the boundary of the parameter space of categorical
response models \citep[see, for example][]{hitchcock:62, gart:67,
  gart:85, clogg:91}.  Especially in cumulative link models where
$g(.)$ is log-concave, if all the counts are positive then the maximum
likelihood estimates cannot be on the boundary of the parameter space
\citep{haberman:80}.

Despite their simplicity, constant adjustment schemes are not
recommended for general use for two reasons:
\begin{enumerate}
\item Because the adjustments are constants,
the resultant estimators are generally not invariant to different
representations of the data (for example, aggregated and disaggregated
view), a desirable invariance property that the maximum likelihood
estimator has, and which allows the practitioner not to be concerned
with whether the data at hand are fully aggregated or not.

\begin{example}
\label{artificial}
For example, consider the two representations of the same data in
Table~\ref{samedata}. Interest is in estimating the difference $\beta$
between logits of cumulative probabilities of the samples with $x =
-1/2$ from the samples with $x = 1/2$.

\begin{table}
\caption{Two alternative representations of the same artificial data
  set.}
\begin{small}
\begin{center}
\begin{tabular*}{0.3\textwidth}{r@{\extracolsep{\fill}}
c@{\extracolsep{\fill}}
c@{\extracolsep{\fill}}
c@{\extracolsep{\fill}}
c@{\extracolsep{\fill}}}
\midrule\midrule
 & \multicolumn{4}{c}{$Y$} \\ \cmidrule{2-5}
$x$ & 1 & 2 & 3 & 4 \\ \midrule
-1/2   & 8 & 6 & 1 & 0 \\
1/2   & 10 & 0 & 1 & 0 \\
1/2   & 8 & 1 & 0 & 0 \\ \midrule\midrule
\end{tabular*}
\hspace{3cm}
\begin{tabular*}{0.3\textwidth}{r@{\extracolsep{\fill}}
c@{\extracolsep{\fill}}
c@{\extracolsep{\fill}}
c@{\extracolsep{\fill}}
c@{\extracolsep{\fill}}}
\midrule\midrule
 & \multicolumn{4}{c}{$Y$} \\ \cmidrule{2-5}
$x$ & 1 & 2 & 3 & 4 \\ \midrule
-1/2   & 8 & 6 & 1 & 0 \\
1/2   & 18 & 1 & 1 & 0 \\ \midrule\midrule
\end{tabular*}
\end{center}
\end{small}
\label{samedata}
\end{table}

The maximum likelihood estimate of $\alpha_3$ is
$+\infty$. Irrespective of the data representation the maximum
likelihood estimate of $\beta$ is finite and has value $-1.944$ with
estimated standard error of $0.895$. Now suppose that the same small
constant, say $1/2$, is added to each of the counts in the rows of the
tables. The adjustment ensures that the parameter estimates are finite
for both representations. Nevertheless, a common constant added to
both tables causes --- in some cases large --- differences in the
resultant inferences for $\beta$. For the left table the maximum
likelihood estimate of $\beta$ based on the adjusted data is $-1.097$
with estimated standard error of $0.678$, and for the right table the
estimate is $-1.485$ with estimated standard error of $0.741$. If
Wald-type procedures were used for inferences on $\beta$ with a Normal
approximation for the distribution of the approximate pivot
$(\hat{\beta} - \beta)/S(\hat\beta)$, where $S(\beta)$ is the
asymptotic standard error at $\beta$ based on the Fisher information,
then the $p$-value of the test $\beta = 0$ would be $0.106$ if the
left table was used and $0.045$ if the right table was used.
\end{example}

\item Furthermore, the moments of the maximum likelihood estimator
  generally depend on the parameter values \citep[see, for
  example][for explicit expressions of the first-order bias term in
  the special case of binomial regression models]{cordeiro:91} and
  thus, as is also amply evident from the studies in
  \citet{hitchcock:62} and \citet{gart:85}, there cannot be a
  universal constant which yields estimates which are optimal
  according to some frequentist criterion.
\end{enumerate}
Both of the above concerns with constant adjustment schemes are dealt
with by using the additive adjustment scheme in
Subsection~\ref{adjustedCounts}. Firstly, by construction, the
iteration of Subsection~\ref{adjustedCounts} yields estimates which
have bias of second-order. Secondly, because the adjustments depend on
the parameters only through the linear predictors which, in turn, do
not depend on the way that the data are represented, the adjustment
scheme leads to estimators that are invariant to the data
representation. For both representations of the data in
Table~\ref{samedata} the bias-reduced estimate of $\beta$ is $-1.761$
with estimated standard error of $0.850$.

\section{Invariance properties of the reduced-bias estimator}
\subsection{Equivariance under linear transformations}
\label{equivariance}
The maximum likelihood estimator is exactly equivariant under
one-to-one transformations ${\bb\phi}(.)$ of the parameter
${\bb\delta}$. That is if $\hat{\bb\delta}$ is the maximum likelihood
estimator of ${\bb\delta}$ then, the maximum likelihood estimator of
${\bb\phi}({\bb\delta})$ is simply ${\bb\phi}(\hat{\bb\delta})$. In
contrast to $\hat{\bb\delta}$, the reduced-bias estimator
$\tilde{\bb\delta}_{RB}$ is not equivariant for all ${\bb\phi}$; bias
is a parameterization-specific quantity and hence any attempt to
improve it can violate exact equivariance. Nevertheless,
$\tilde{\bb\delta}_{RB}$ is equivariant under linear transformations
${\bb\phi}({\bb\delta}) = L{\bb\delta}$, where $L$ is a $(p+q) \times
(p+q)$ matrix of constants such that $ZL$ is of full rank and
${\bb\delta}' =L{\bb\delta}$ has $\alpha'_1 < \ldots < \alpha'_q$.

To see that, assume that one fits the multinomial model with
$\gamma_{rs} = G(\eta_{rs}')$ where $\eta_{rs}' = \sum_{t = 1}^{p +
  q}{\bb\delta}_t' z_{rst}$ $(r = 1, \ldots, n, s = 1, \ldots, q)$. Because
${\bb\delta}' = L{\bb\delta}$, $\eta_{rs}'$ is a linear combination of
${\bb\delta}$. Using expression~(\ref{adjustedScore2}), the $t$th component
of the adjusted score function for ${\bb\delta}'$ is
\begin{equation}
\label{dashedAdjScore}
U_t' = \sum_{r = 1}^n\sum_{s = 1}^q g_{rs}' \left\{\frac{y_{rs} +
    c_{rs}' - c_{rs-1}'}{\pi_{rs}'} - \frac{y_{rs+1} + c_{rs+1}' -
    c_{rs}')}{\pi_{rs+1}'}
\right\}z_{rst} \, ,
\end{equation}
for $t \in \{1, \ldots, p + q\}$, where $c_{rs}'$, $\pi_{rs}'$,
$g_{rs}'$ are evaluated at ${\bb\delta}'$. Note that all quantities in
(\ref{dashedAdjScore}) depend on ${\bb\delta}'$ only though the linear
combinations $\eta_{rs}'$. Thus, comparing (\ref{adjustedScore2}) to
(\ref{dashedAdjScore}), if $\tilde{\bb\delta}_{RB}$ is a solution of
$U_t^* = 0$ $(t = 1, \ldots, p+q)$, then $L\tilde{\bb\delta}_{RB}$ must be
a solution of $U_t' = 0$ $(t = 1, \ldots, p+q)$.

The bias-corrected estimator defined in (\ref{biasCorrection}) can be
shown also to be equivariant under linear transformations, using the
equivariance of the maximum likelihood estimator and the fact that
${\bb b}({\bb\delta})$ depends on ${\bb\delta}$ only through the
linear predictors.

\subsection{Invariance under reversal of the order of categories}
One of the properties of proportional-odds models, and generally of
cumulative link models with a symmetric latent distribution $G(.)$ is
their invariance under the reversal of the order of categories; a
reversal of the categories along with a simultaneous change of the
sign of ${\bb\beta}$ and change of sign --- and hence order --- to
$\alpha_1, \ldots, \alpha_q$ in model~(\ref{model}) results in the
same category probabilities. Given the usual arbitrariness in the
definition of ordinal scales in applications this is a desirable
invariance property for the analysis of ordinal data.

The maximum likelihood estimator respects this invariance
property. That is if the categories are reversed then the new fit can
be obtained by merely using $-\hat{\bb\beta}_{ML}$ for the regression
parameters and $(-\hat\alpha_{q}, \ldots, -\hat\alpha_{1})$ for the
cutpoints.

The reduced-bias estimator respects the same invariance property,
too. To see this, assume that one fits the multinomial model with $1 -
\gamma_{rs} = G(\alpha_{k - s} - {\bb\beta}^T{\bb x}_r)$ $(r = 1,
\ldots, n, s = 1, \ldots, q)$ with $\alpha_{1} < \ldots <
\alpha_{q}$. Because $g(.)$ is symmetric about zero, $G(\eta) = 1 -
G(-\eta)$, and so $\gamma_{rs} = G(-\alpha_{k - s} +
{\bb\beta}^Tx_r)$. This is a reparameterization of
model~(\ref{glmForm}) to $\gamma_{rs} = G(\sum_{t =
  1}^{p+q}{\bb\delta}_t'z_{rst})$ where ${\bb\delta}' = (\alpha_1',
\ldots, \alpha_q', \beta_1', \ldots, \beta_p')^T = (-\alpha_q, \ldots,
-\alpha_1, -\beta_1, \ldots, -\beta_p)^T$. Hence, ${\bb\delta}' = L
{\bb\delta}$ with
\[
L = \left[
\begin{array}{rrrrr}
0 & \ldots & 0 & -1 & 0\\
0 & \ldots & -1 & 0 & 0\\
\vdots & \iddots & \vdots & \vdots & \vdots \\
-1 & \ldots & 0 & 0 & 0 \\
0 & \ldots & 0 & 0 & -1 \\
\end{array}
\right] \,,
\]
and based on the results of Subsection~\ref{equivariance},
$\tilde{\bb\delta}'_{RB} = L \tilde{\bb\delta}_{RB}$ (and also
$\tilde{\bb\delta}'_{BC} = L \tilde{\bb\delta}_{BC}$).

\section{Properties of the reduced-bias estimator and associated
  inferential procedures: a complete enumeration study}
\label{completeenumeration}
\subsection{Study design}
The frequentist properties of the reduced-bias estimator are
investigated through a complete enumeration study of $2 \times k$
contingency tables with fixed row totals.  The rows of the tables
correspond to a two-level covariate $x$ with values $x_1$ and $x_2$,
and the columns to the levels of an ordinal response $Y$ with
categories $1, \ldots, k$. The row totals are fixed to $m_1$ for $x =
x_1$ and to $m_2$ for $x = x_2$. The right table in
Table~\ref{samedata} is a special case with $k = 4$, $x_1 = -1/2$,
$x_2 = 1/2$, and row totals $m_1 = 15$, $m_2 = 20$.  The present
complete enumeration involves $\binom{m_1 + q}{m_1}\binom{m_2 +
  q}{m_2}$ tables. We consider a multinomial model with
\begin{align}
\label{2xk_model}
\gamma_{1s} & = G(\alpha_s - \beta x_1)\, ,\\ \notag
\gamma_{2s} & = G(\alpha_s - \beta x_2) \quad (s=1,\ldots,q)\, ,
\end{align}
where $\alpha_1, \ldots, \alpha_q$ are regarded as nuisance parameters
but are essential to be estimated from the data, because they allow
flexibility in the probability configurations within each of the rows
of the table.

For the estimation of $\beta$ we consider the maximum likelihood
estimator $\hat\beta$, the bias-corrected estimator
$\tilde\beta_{BC}$, the reduced-bias estimator $\tilde\beta_{RB}$, and
the generalized empirical logistic transform $\hat\beta_{EL}$ which is
defined in \citet[][\S 2.3]{mccullagh:80} and is an alternative
estimator with smaller asymptotic bias than the maximum likelihood
estimator specifically engineered for the estimation of $\beta$ in
$2\times k$ tables with fixed row totals. The estimators $\hat\beta$,
$\tilde\beta_{BC}$ and $\tilde\beta_{RB}$ are the $\beta$-components
of the vectors of estimators $\hat{\bb\delta}$, $\tilde{\bb\delta}_{BC}$ and
$\tilde{\bb\delta}_{RB}$, respectively, where ${\bb\delta} = (\alpha_1, \ldots,
\alpha_q, \beta)^T$ is the vector of all parameters. The estimators
are compared in terms of bias, mean-squared error and coverage
probability of the respective Wald-type asymptotic confidence
intervals. The following theorem is specific to $2\times k$ and
cumulative link models, and can be used to reduce the parameter
settings that need to be considered in the current study for
evaluating the performance of the estimators.

\begin{theorem}
\label{symmetry}
  Consider a $2 \times k$ contingency table $T$ with fixed row totals
  $m_1$ and $m_2$, and the multinomial model that satisfies
  (\ref{2xk_model}). Furthermore, consider an estimator ${\bb\delta}^*(T)$
  of ${\bb\delta}$, which is equivariant under linear transformations. Then
  if $m_1 = m_2$, the bias function and the mean squared error of
  $\beta^*(T)$ satisfy
  \[
  E(\beta^*(T) - \beta; \beta, {\bb\alpha}) =  -E(\beta^*(T) + \beta;
  -\beta, {\bb\alpha}) \, , \quad \text{and}
  \]
  \[
  E\left\{(\beta^*(T) - \beta)^2; \beta, {\bb\alpha}\right\} =
  E\left\{(\beta^*(T) + \beta)^2; -\beta, {\bb\alpha}\right\} \, ,
  \quad \text{respectively}\,.
  \]
\end{theorem}
\begin{proof}
  Define an operator $R$ which when applied to $T$ results in a new
  contingency table by reversing the order of the rows of $T$. Hence,
  $R(R(T)) = T$.

  Because ${\bb\delta}^*(T)$ is equivariant under linear
  transformations, it suffices to study the behaviour of $\beta^*(T)$
  when $x_1 = -1/2$ and $x_2 = 1/2$. Then, any combination of values
  for $x_1$ and $x_2$ results by an affine transformation of the
  vector $(-1/2, 1/2)$, and equivariance gives that a corresponding
  translation of the vector ${\bb\alpha}^*(T)$ and change of scaling
  of $\beta^*(T)$ results in exactly the same fit. Hence, the shape
  properties of $\beta^*(T)$ remain invariant to the choice of $(x_1,
  x_2)^T$.

  Denote with $\mathcal T$ the set of all possible $2\times k$ tables
  with fixed row totals $m_1$ and $m_2$. By the definition of the
  model, $P(T; \beta, {\bb\alpha}) = P(R(T); -\beta, {\bb\alpha})$ for
  every $T \in \mathcal T$.  Because $m_1 = m_2$ there is a subset
  $\mathcal E \subset \mathcal T$ of tables with $(y_{11}, \ldots,
  y_{1k}) = (y_{21}, \ldots, y_{2k})$.  The complement of $\mathcal E$
  can be partitioned into the sets $\mathcal F_1$ and $\mathcal F_2$
  which have the same cardinality, and where $T \in \mathcal F_1$ if
  and only if $R(T) \in \mathcal F_2$. For $x_1 = -1/2$ and $x_2 =
  1/2$, equivariance under the linear transformation $\phi(\beta) =
  -\beta$ gives that $\beta^*(T) = -\beta^*(R(T))$.  Then, for any $T
  \in \mathcal E$, $\beta^*(T) = 0$. Hence,
  \begin{align}
    \label{calc}
    E(\beta^*(T); \beta, {\bb\alpha}) & = \sum_{T \notin \mathcal E}
    \beta^*(T) P(T; \beta, {\bb\alpha}) \\ \notag
    & = \sum_{T \notin \mathcal E, T
      \in \mathcal F_1} \beta^*(T) \left\{ P(T;
      \beta, {\bb\alpha}) - P(R(T); \beta, {\bb\alpha})\right\} \\ \notag
    & = \sum_{T \notin \mathcal E, T \in \mathcal F_1} \beta^*(T)
    \left\{
      P(R(T); -\beta, {\bb\alpha}) - P(T; -\beta, {\bb\alpha})\right\}
    = -E(\beta^*(T); -\beta, \alpha)
  \end{align}
  Adding $-\beta$ to both sides of the above equality gives the
  identity on the bias. For the identity on the mean squared error one
  merely needs to repeat a corresponding calculation to (\ref{calc})
  starting from $E\{(\beta^*(T) - \beta)^2; \beta, {\bb\alpha}\} = \sum_{T
    \notin \mathcal E} (\beta^*(T) - \beta)^2 P(T; \beta, {\bb\alpha}) +
  \beta^2 \sum_{T \in \mathcal E} P(T; \beta, {\bb\alpha})$.
\end{proof}

A similar line of proof can be used to show that if $m_1 = m_2$ the
coverage probability of Wald-type asymptotic confidence intervals for
$\beta$ is symmetric about $\beta = 0$, provided that the estimator
$S(T)$ of the standard error of $\beta^*(T)$ satisfies $S(T) =
S(R(T))$.

\subsection{Special case: Proportional odds model}
For demonstration purposes, the values of the competing estimators are
obtained for a proportional odds model ($G(\eta) = \exp(\eta)/\{1 +
\exp(\eta)\}$) with $x_1 = -1/2$ and $x_2 = 1/2$ and $k = 4$, for each
of the $400$, $3136$ and $81796$ possible tables with row totals $m =
m_1 = m_2$, for $m = 3$, $m =5$, and $m = 10$, respectively. All
estimators considered are equivariant under linear transformations and
hence, according to the proof of Theorem~\ref{symmetry}, the outcome
of the complete enumeration for the comparative performance of the
estimators generalizes to any choice of $(x_1, x_2)^T$.

The estimators $\hat\beta$ and $\tilde\beta_{RB}$ are not available in
closed form and one needs to rely on iterative procedures for finding
the roots of $U_t({\bb\delta})$ and $U_t^*({\bb\delta})$, respectively, for
every $t \in \{1, 2, 3, 4\}$. Fisher scoring is used to obtain
$\hat\beta$ and the iterative maximum likelihood approach of
Subsection~\ref{adjustedCounts} is used for $\tilde\beta_{RB}$. The
maximum likelihood estimate is judged satisfactory if the current
value ${\bb\delta}^c$ of the iterative algorithm satisfies $|U_t({\bb\delta}^c)|
< 10^{-10}$ for every $t \in \{1, 2, 3, 4\}$. For $\tilde\beta_{RB}$,
the latter criterion is used with $U_t^*$ in the place of $U_t$.

\begin{figure}[t!]
\caption{A pictorial representation of the probability settings
  considered in the calculation of expectations from the complete
  enumeration study. The left hand side of each plot depicts the
  multinomial probabilities for $x = -1/2$ and the right the
  multinomial probabilities for $x = 1/2$. The $8$ probabilities ($4$
  for each $x$ value) for each particular combination of values for
  $\beta$ and $(\alpha_1, \alpha_2, \alpha_3)$ are connected with line
  segments. Hence each piecewise linear function on each plot
  corresponds to a specific probability setting for the $2\times 4$
  contingency table with fixed row totals. The plots correspond to
  particular settings for the nuisance parameters $(\alpha_1,
  \alpha_2, \alpha_3)$ determined by $e (-1, 0, 1)$, and each plot
  contains all possible piecewise linear functions for the values of
  $\beta$ on an equi-spaced grid of size $50$ in the interval $[-6,
  6]$.}
\includegraphics[width = \textwidth]{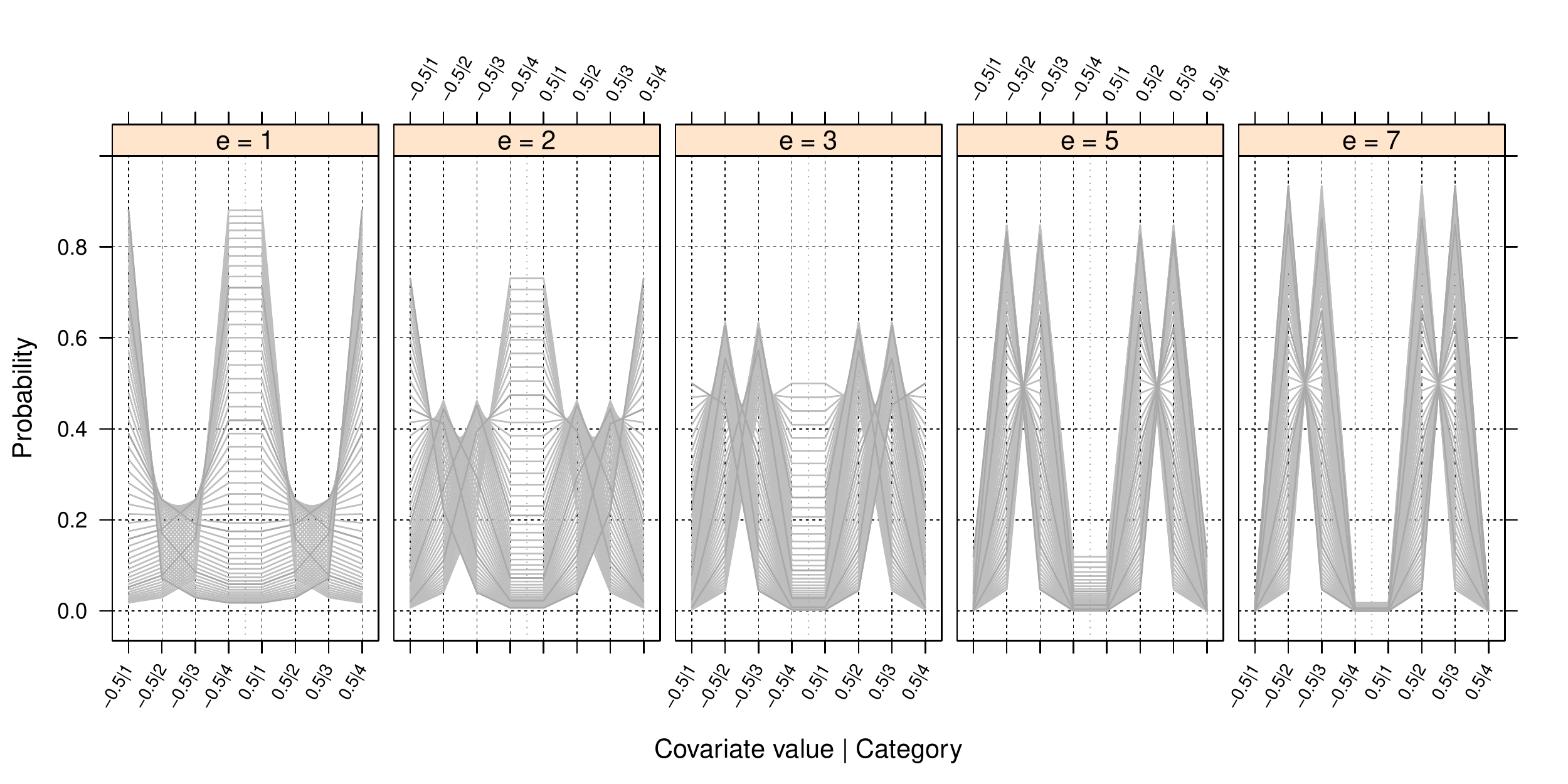}
\label{probsettings}
\end{figure}

For evaluating the performance of the estimators, the probability of
each of the tables has been calculated under model~(\ref{2xk_model}),
for parameter values that are fixed according to the following
scheme. The parameter $\beta$ takes values on some sufficiently fine
equi-spaced grid in the interval $[-6, 0]$. For $\beta$ in the
interval $(0, 6]$ the results can be predicted by the symmetry
relations of Theorem~\ref{symmetry}. For each value of $\beta$, the
nuisance parameters take values $(\alpha_1, \alpha_2, \alpha_3)^T = e
(-1, 0, 1)^T$ for $e \in \{1, 2, 3, 5, 7\}$. Figure~\ref{probsettings}
is a pictorial representation of the probability settings for the two
multinomial vectors in the $2 \times 4$ contingency table with fixed
row totals, at each combination of values for $\beta$ and $(\alpha_1,
\alpha_2, \alpha_3)^T$. Under the above scheme for fixing parameter
values, the probability of the end categories tends to zero as $e$
increases, and hence more extreme probability settings are being
considered as $e$ grows.

The findings of the current complete enumeration exercise are outlined
in the following Subsection.  The same complete enumeration design has
been applied to a number of settings, with $m_1 \ne m_2$, with
different link functions, with different numbers of categories, and/or
for different non-symmetric specifications for the nuisance parameters
(results not shown here) yielding qualitatively the same conclusions;
the current setup merely allows a clear pictorial representation of
the findings on the behaviour of the reduced-bias estimator. An R
script that can produce the results of the current complete
enumeration for any number of categories, any link function, any
configuration of totals and any combination of parameter settings in
$2 \times k$ contingency tables is available in the supplementary
material.

\begin{figure}[t!]
\caption{Probability of infinite estimates (top), conditional biases
  (middle) and conditional mean squared errors (bottom) of
  $\hat\beta$ and $\tilde\beta_{BC}$ for the parameter settings
  considered in the complete enumeration study.}
\begin{center}
\includegraphics[height = 0.27\textheight]{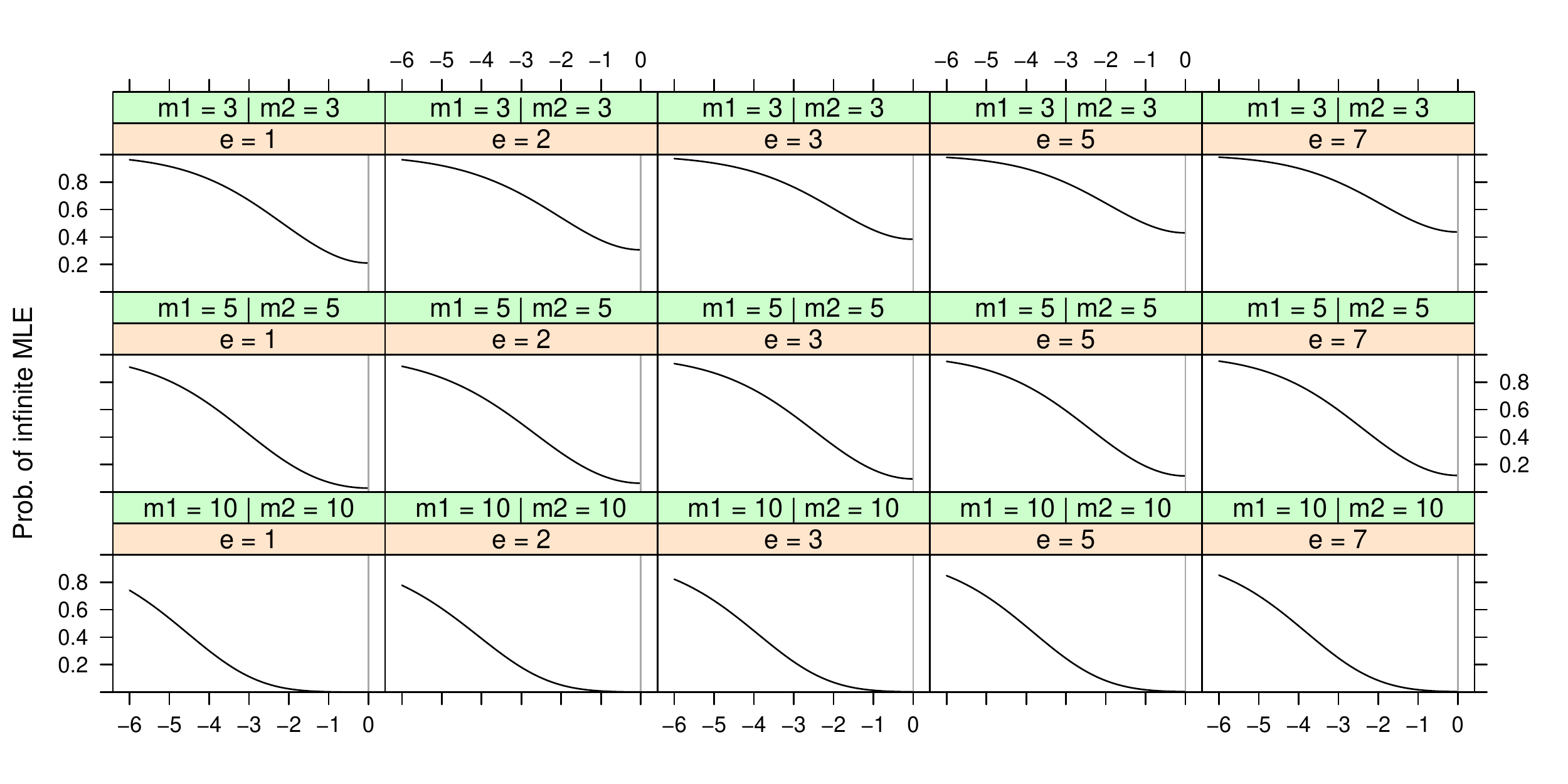}
\includegraphics[height = 0.295\textheight]{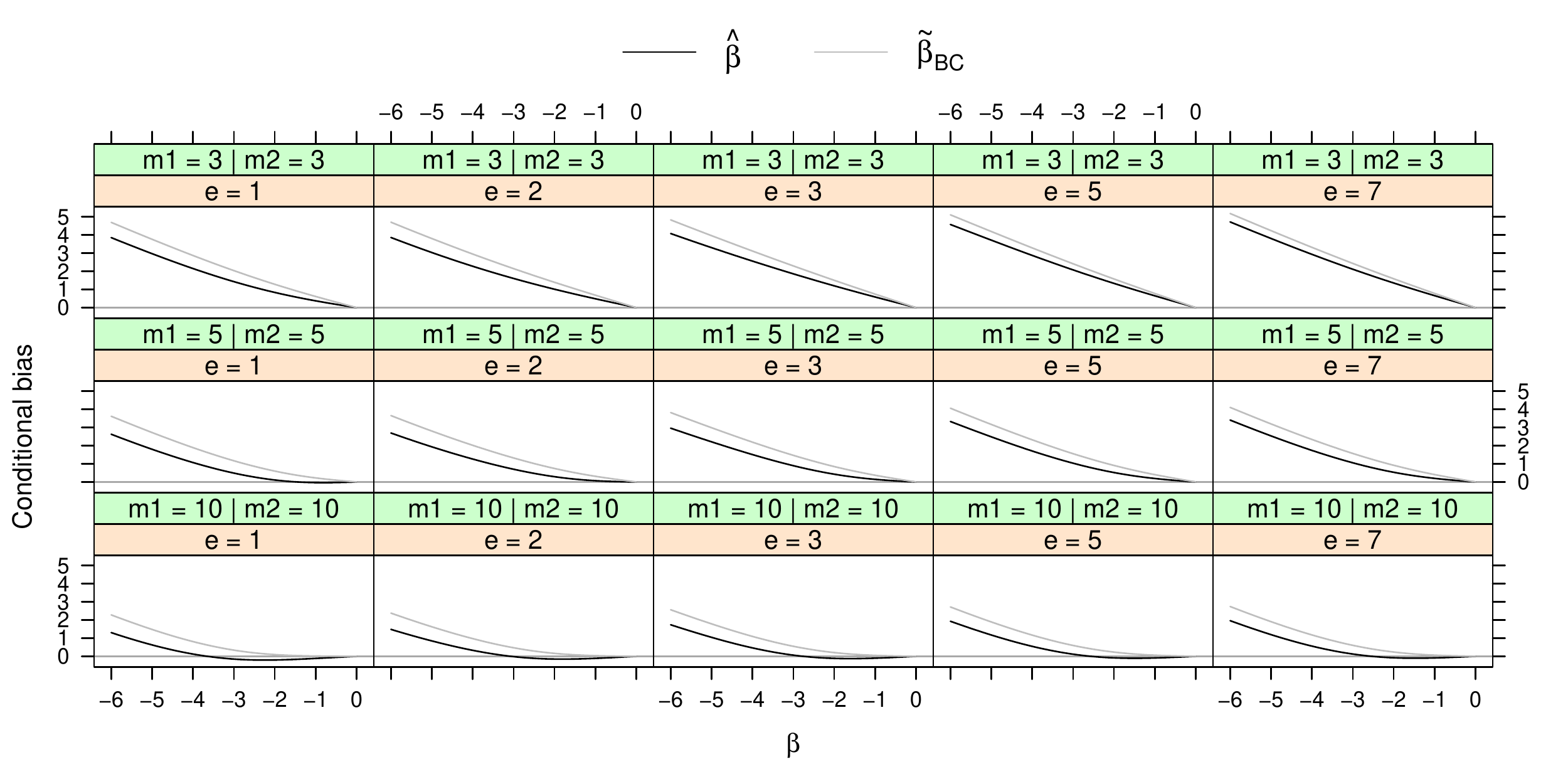}
\includegraphics[height = 0.295\textheight]{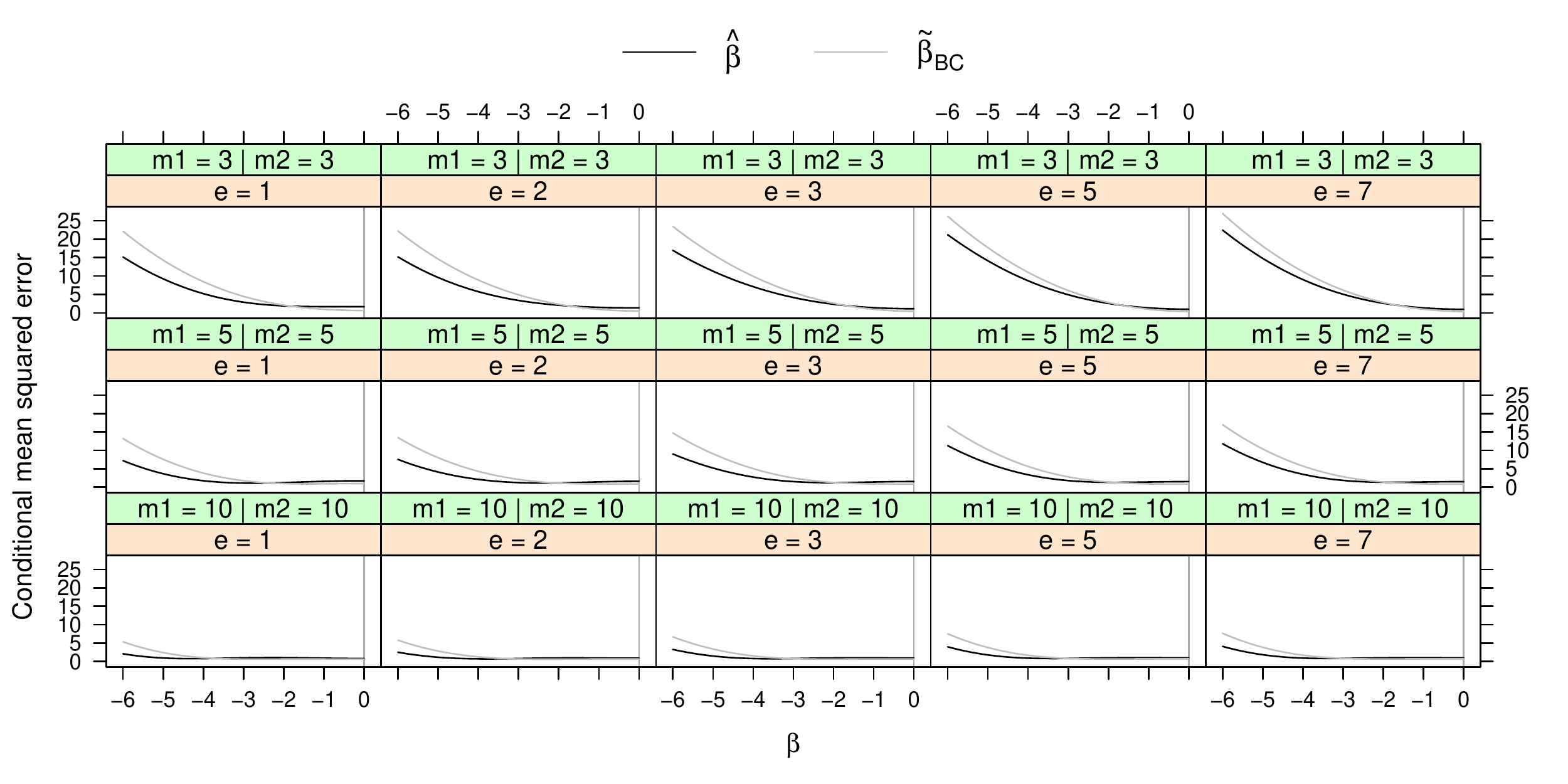}
\end{center}
\label{mlplots}
\end{figure}

\subsection{Remarks on the results}
\noindent {\bf Remark 1. On the estimates of $\alpha_1$, $\alpha_2$
  and $\alpha_3$:} According to Section~\ref{infinite}, for data sets
where a specific category $s \in \{1, 2, 3, 4\}$ is observed for
neither $x = -1/2$ nor $x = 1/2$, the maximum likelihood estimate of
$\alpha$ is on the boundary of the parameter space as follows:
\[
\begin{array}{ll}
  s = 1: & \hat\alpha_1 = -\infty  \\
  s = 2: & \hat\alpha_2 = \hat\alpha_1 \\
  s = 3: & \hat\alpha_3 = \hat\alpha_2 \\
  s = 4: & \hat\alpha_3 = +\infty \\
\end{array} \, .
\]

A least for log-concave $g(.)$, according to the results in
\citet{pratt:81}, the above equations extend directly to the case of
any number of categories and number of covariate settings and can
directly be used to check what happens when two or more categories are
unobserved.

Nevertheless, the maximum likelihood estimator of $\beta$ is invariant
to merging a non-observed category with either the previous or next
category and can be finite even if some of the $\alpha$ parameters are
on the boundary of the parameter space. Hence, maximum likelihood
inferences on $\beta$ are possible even if a category is not
observed. The same behaviour is observed for the reduced-bias
estimators of $\alpha_1, \alpha_2, \alpha_3$ with the difference that
if the non-observed category is $s = 1$ and/or $s = 4$, then
$\tilde\alpha_{1,RB}$ and/or $\tilde\alpha_{3,RB}$ are finite. A
special case of this observation has been encountered in
Subsection~\ref{cumlogits} where reduction of the bias corresponds to
adding $1/2$ to the end categories, guaranteeing the finiteness of the
cumulative logits. Hence, there is no need for non-observed
end-categories to be merged with the neighbouring ones when the
reduced-bias estimator is used. If any of the other categories is
empty, then the reduced-bias estimator of $\beta$ is invariant to
merging those with any of the neighbouring ones.

It should be mentioned here that if both the second and the third
category are empty then the reduced-bias estimate of $\beta$ and the
generalized empirical logistic transform are identical. To see that,
note that in the special case of logistic regression, the adjusted
scores in Subsection~\ref{binomModels} suggest adding half a leverage
to each of $y_{r1}$ and $y_{r2}$ $(r = 1, 2)$ \citep[this result for
logistic regressions was obtained in][]{firth:93}. Furthermore, the
model with $q=1$ is saturated and hence both leverages are $1$. Hence
the reduced-bias estimate of $\beta$ coincides with the generalized
empirical logistic transform, which for $k =2$ is $\log\{(y_{11} +
1/2)/(m_1 - y_{11} + 1/2)\} - \log\{(y_{21} + 1/2)/(m_2 - y_{21} +
1/2)\}$.

\bigskip\noindent {\bf Remark 2. On $\hat\beta$ and
  $\tilde\beta_{BC}$:} As is expected from the discussion in
Section~\ref{infinite}, the maximum likelihood estimator of $\beta$ is
infinite for certain configurations of zeros on the table, and for
such configurations the bias-corrected estimator is also undefined
owing to its explicit dependence on the maximum likelihood
estimator. Hence, for $\hat\beta$ and $\tilde\beta_{BC}$, the bias
function is undefined and the mean squared error is infinite. A
possible comparison of the performance of $\hat\beta$ and
$\tilde\beta_{BC}$ is in terms of conditional bias and conditional
mean squared error where the conditioning event is that $\hat\beta$
has a finite value.

For detecting parameters with infinite values the diagnostics in
\citet[][\S 4]{lesaffre:89} for multinomial logistic regressions are
adapted to the current setting. Data sets that result in infinite
estimates for $\beta$ have been detected by observation of the size of
the corresponding estimated standard error based on the inverse of the
Fisher information, and by observation of the absolute value of the
estimates when the convergence criteria were satisfied. If the
standard error was greater than $200$ and the estimate was greater
than $100$, then the estimate was labelled infinite. A second pass
through the data sets has been performed making the convergence
criterion for the Fisher scoring stricter than $|U_t({\bb\delta}^c)| <
10^{-10}$. The estimates that were labelled infinite using the
aforementioned diagnostics, further diverged towards infinity while
the rest of the estimates remained unchanged to high accuracy.

The probability of encountering an infinite $\hat\beta$ for the
different possible parameter settings is shown at the top row of
Figure~\ref{mlplots}. For $\beta \in (0, 6)$ the probability of
encountering an infinite value is simply a reflection of the
probability in $( -6, 0)$. As is apparent the probability of infinite
estimates increases as $e$ increases and for each value of $e$ it
increases as $|\beta|$ increases. As is natural as $m$ increases, the
probability of encountering infinite estimates is reduced but is
always positive.

Of course, the findings from the current comparison of $\hat\beta$
with $\tilde\beta_{BC}$ should be interpreted critically, bearing in
mind the conditioning on the finiteness of $\hat\beta$; the comparison
suffers from the fact that the first-order bias term that is required
for the calculation of $\tilde\beta_{BC}$ is calculated
unconditionally. The comparison is fairer when the probability of
infinite estimates is small; this happens on a region around zero
whose size also increases as $m$ increases.

The conditional bias and conditional mean squared error of $\hat\beta$
and $\tilde\beta_{BC}$ are shown in the left and right of the second
row of Figure~\ref{mlplots}. The identities in Theorem~\ref{symmetry}
apply also for the conditional and conditional mean squared error; to
see this set $P$ to be the conditional probability of each table in
the proof of Theorem~\ref{symmetry}. Hence, for $\beta\in (0, 6)$, the
conditional bias is simply a reflection of the conditional bias for
$\beta\in (-6, 0)$ across the $45^o$ line, and the conditional mean
squared error is a reflection of the conditional mean squared error for
$\beta \in (-6, 0)$ across $\beta = 0$.

\begin{figure}[t!]
\caption{Biases (top) and mean squared errors (bottom) of
  $\hat\beta_{EL}$ and $\tilde\beta_{RB}$ for the parameter settings
  considered in the complete enumeration study.}
\begin{center}
\includegraphics[height = 0.295\textheight]{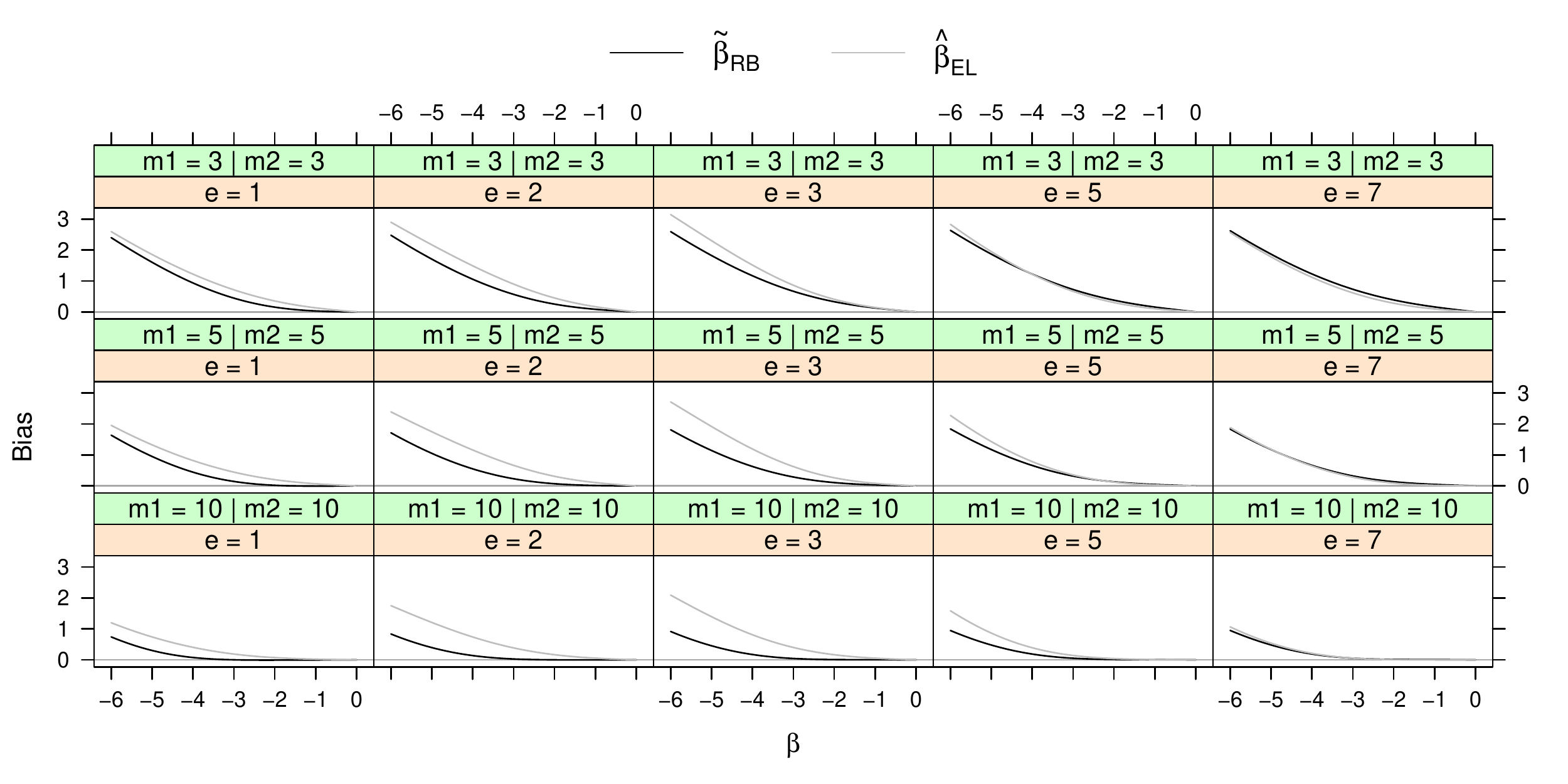}
\includegraphics[height = 0.295\textheight]{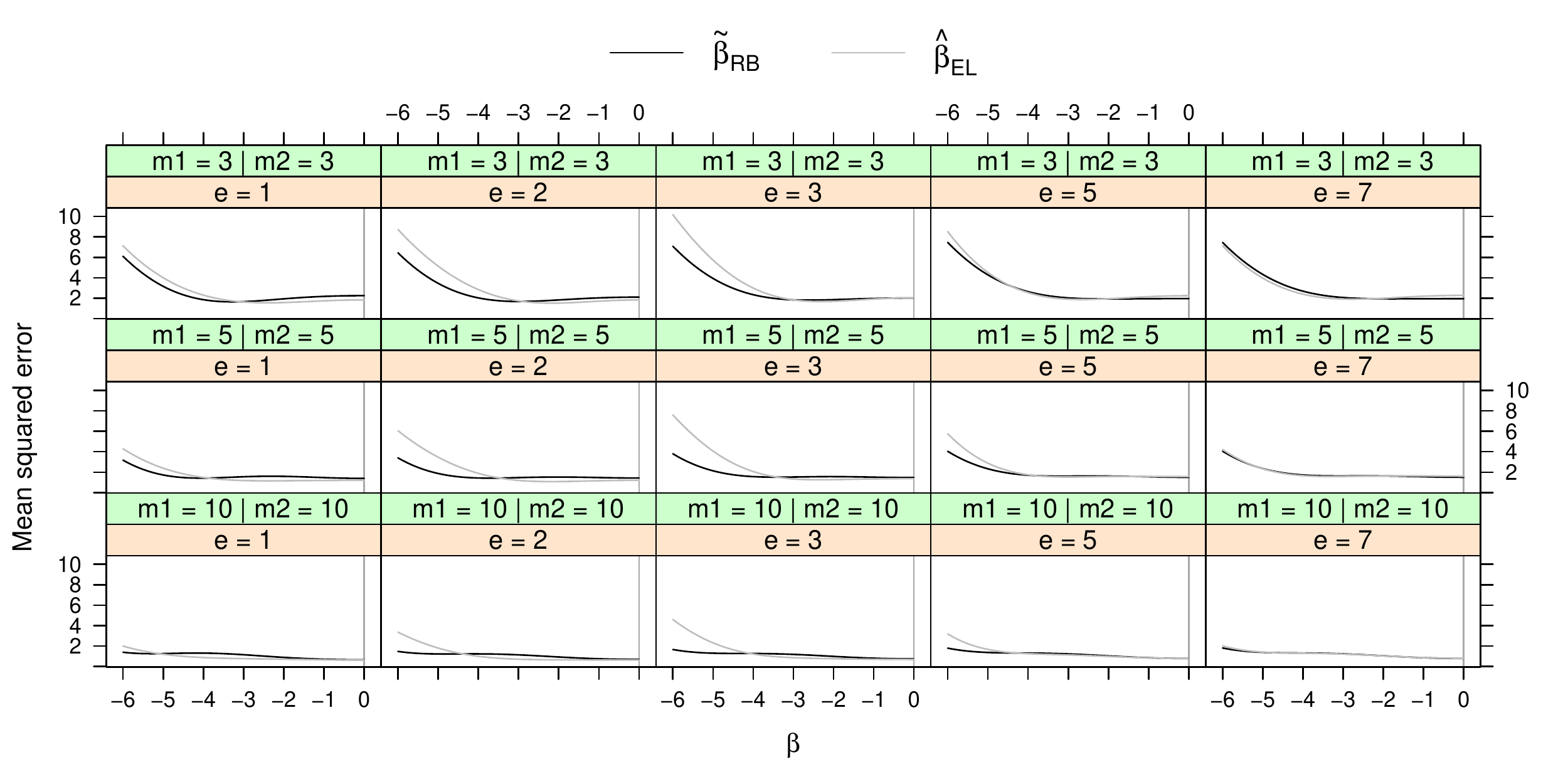}
\end{center}
\label{brplots}
\end{figure}

The behaviour of the estimators in terms of conditional bias is
similar, with the maximum likelihood estimator performing slightly
better than $\tilde\beta_{BC}$ for small $m$. As $m$ increases the bias
corrected estimator starts performing better in terms of bias in a
region around zero, where the probability of infinite estimates is
smallest. The same is noted for the conditional mean squared
error. The estimator $\tilde\beta_{BC}$ performs better than
$\hat\beta$ in a region around zero, whose size increases as $m$
increases. The same behaviour as for $e = 7$ persists for larger
values of $e$ (figures not shown here).

\bigskip\noindent {\bf Remark 3. On $\hat\beta_{EL}$ and
  $\tilde\beta_{RB}$:} The estimators $\hat\beta_{EL}$ and
$\tilde\beta_{RB}$, always have finite value irrespective of the
configuration of zeros on the table. Hence, in contrast to $\hat\beta$
and $\tilde\beta_{BC}$, a comparison in terms of their unconditional
bias and unconditional mean squared error is possible. The left plot
of Figure~\ref{brplots} shows the bias function of the estimator for
the parameter settings considered in the complete enumeration
study. For $\beta\in (0, 6)$, the bias function is simply a reflection
of the bias for $\beta\in (-6, 0)$ across the $45^o$ line, and the
mean squared error is a reflection of the mean squared error for
$\beta \in (-6, 0)$ across $\beta = 0$.

The reduced-bias estimator performs better than $\hat\beta_{EL}$ in
terms of bias for small values of $e$ and the differences in the bias
functions diminish as $e$ increases. A similar limiting behaviour
holds for their mean squared errors, though for small values of $e$,
$\hat\beta_{EL}$ performs slightly better than $\tilde\beta_{BR}$ in
terms of mean squared error in the range $(-4, 4)$ and worse outside
that range. The mean squared error of both estimators converges to
zero as $m$ increases, which is what is expected from consistent
estimators \citep[see,][\S 6.3 for a proof of the consistency of the
reduced-bias estimator]{kosmidis:07}.

\begin{figure}[t]
\caption{Coverage probabilities of nominally 95\% asymptotic
  Wald-type confidence intervals for $\beta$, based on $\hat\beta$ and
  $\tilde\beta_{BC}$ (top) and $\hat\beta_{EL}$ and
  $\tilde\beta_{RB}$ (bottom) and the respective standard errors, for
  $\beta \in [-10, 0)$ and $\alpha = e (-1, 0, 1)^T$ for $e \in \{1, 2, 3,
  5, 7\}$.}
\begin{center}
\includegraphics[height = 0.295\textheight]{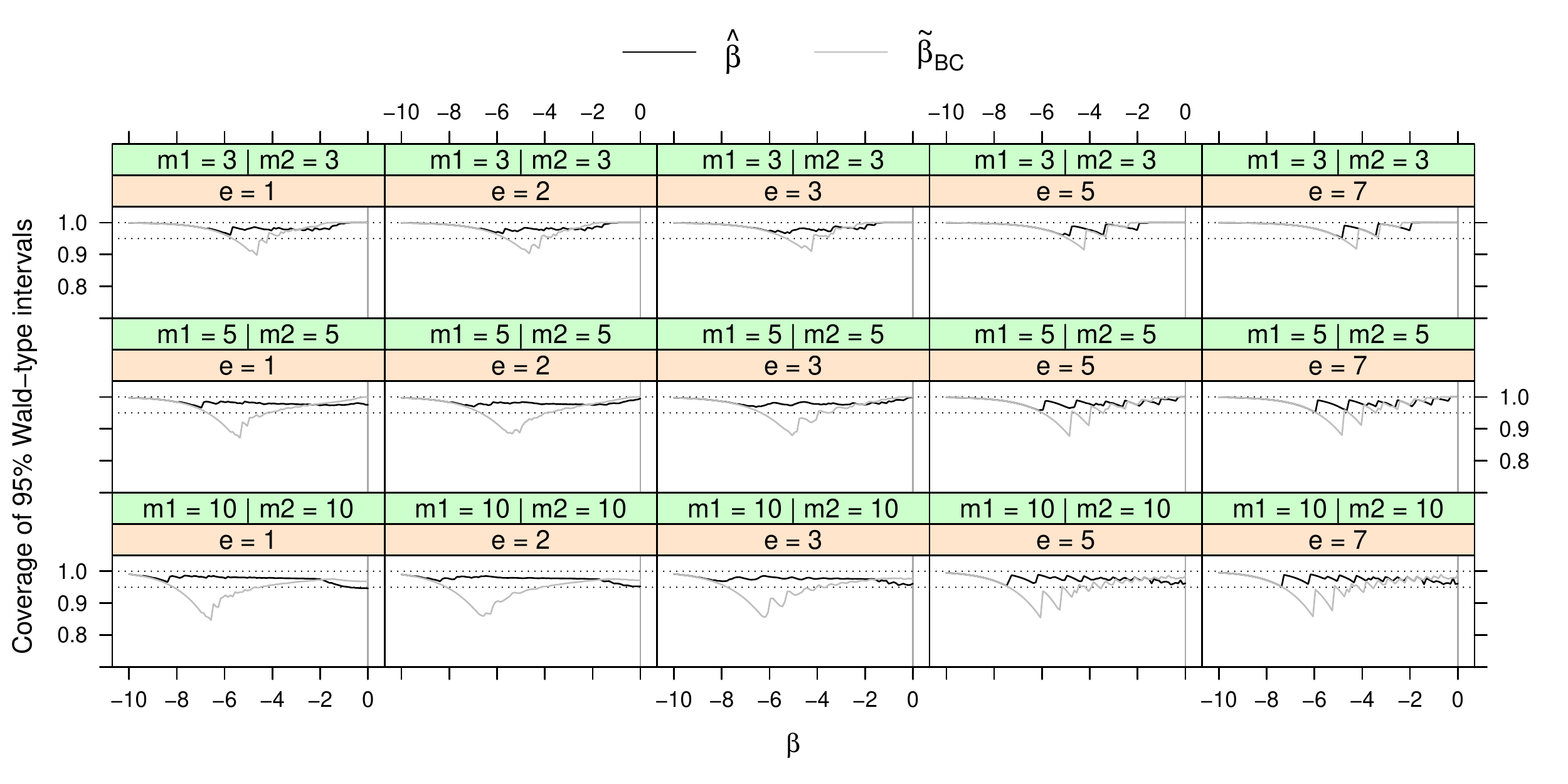}
\includegraphics[height = 0.295\textheight]{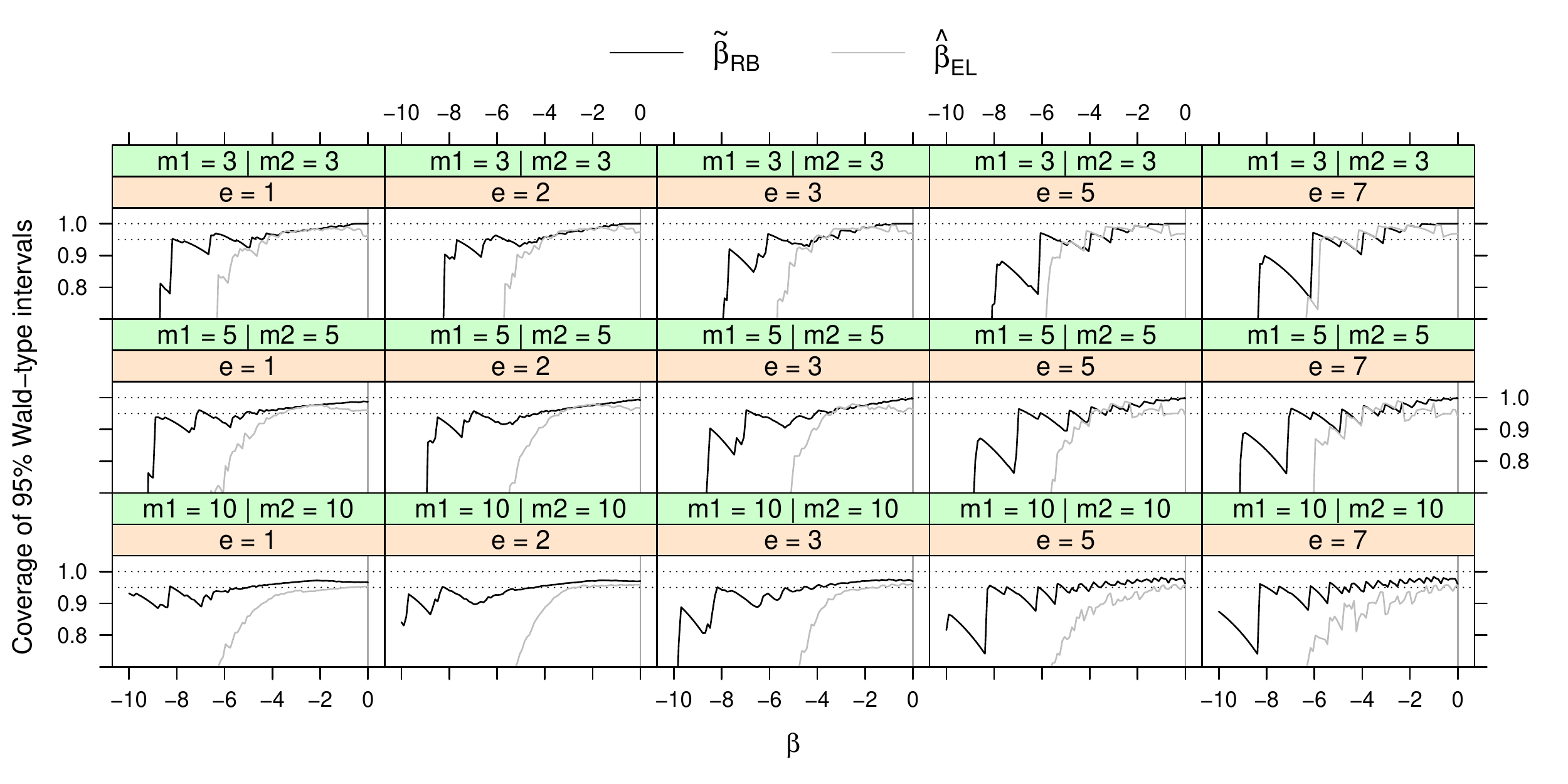}
\end{center}
\label{coverplots}
\end{figure}

\bigskip\noindent {\bf Remark 4. On the coverage of 95\% Wald
  confidence intervals} For a table $T$ and an estimator
$\beta^*(T)$, consider the nominally $100(1-a)\%$ Wald-type
confidence interval for $\beta$
\[
\beta^*(T) \pm z_{1 - a/2} S^*(T) \, ,
\]
where $z_{a}$ is the $100a$th quantile of a standard normal
distribution and $S^*(T)$ is the estimator of the standard error of
$\beta^*(T)$. For $\hat\beta$, $\tilde\beta_{BC}$ and
$\tilde\beta_{RB}$, $S^*(T)$ is taken to be the square root of the
diagonal element of the inverse of the Fisher information
corresponding to $\beta$, evaluated at $\hat\beta(T)$,
$\tilde\beta_{BC}(T)$ and $\tilde\beta_{RB}(T)$, respectively. For the
estimation of the standard error for $\hat\beta_{EL}$, the variance
formula given in \citet[\S 2.3]{mccullagh:80} is used. If the maximum
likelihood estimate is infinite then we make the convention that the
confidence intervals based on $\hat\beta$ and $\tilde\beta_{BC}$ are
$(-\infty, \infty)$. Figure~\ref{coverplots} shows the coverage
probabilities of the four competing intervals for $\alpha = e (-1, 0,
1)^T$ with $e \in \{1, 2, 3, 5, 7\}$, and for $\beta \in [-10,
0]$. The coverage probability for $\beta \in (0, 10)$ is simply a
reflection of the coverage probability for $\beta \in (-10, 0)$ across
$\beta = 0$.

Wald-type confidence intervals based on the maximum likelihood
estimator demonstrate a conservative behaviour in terms of coverage,
and the coverage probability converges to $1$ as $|\beta| \to
\infty$. Furthermore, the coverage probability seems to uniformly get
closer to the nominal level as $m$ increases. The intervals based on
the bias-corrected estimator also demonstrate a conservative behaviour
in a neighbourhood around $\beta = 0$, then tend to undercover for an
interval of large $|\beta|$ values, and as for $\hat\beta$, when
$|\beta| \to \infty$ the coverage probability tends to $1$.

A more dramatic undercoverage is present for confidence intervals
based on $\hat\beta_{EL}$ when $|\beta|$ is large. Actually after some
value of $|\beta|$ the confidence intervals based on $\hat\beta_{EL}$
completely lose coverage (the full range of the coverage probability
is not shown here). On the other hand, those intervals behave
satisfactorily around $\beta = 0$. This behaviour relates to the fact
that the variance estimator for $\hat\beta_{EL}$ is obtained under the
assumption that $\beta = 0$ and can seriously underestimate the
variance of $\hat\beta_{EL}$ when $|\beta|$ is larger than about $1$
\citep[the same observation is also made in][\S
2.3]{mccullagh:80}. Furthermore, it is worth noting that the point
where coverage is lost completely moves closer to zero as $m$
increases. Hence, use of Wald-type confidence intervals based on
$\hat\beta_{EL}$ is not recommended in practical applications.

Apart from being conservative, confidence intervals based on
$\tilde\beta_{RB}$ seem to behave better for a wider range of $\beta$
around zero, but also completely lose coverage after some value of
$|\beta|$. The complete loss of coverage for large effects is due to
an interplay of discreteness of the response and the fact that
$\tilde\beta_{RB}$ and $\hat{\beta}_{EL}$ take always finite values.
Specifically, for any finite $m$ there is only a finite number of
possible Wald-type confidence intervals because the response is
multinomially distributed, and any of those confidence intervals has
finite endpoints. Therefore, there will always be a large enough value
of $|\beta|$ which is not contained in any of the confidence intervals
resulting to a complete loss of coverage. Nevertheless, in contrast to
$\hat\beta_{EL}$, the coverage properties of the Wald-type confidence
intervals based on $\hat\beta_{RB}$ improve quickly and the value
where coverage is lost moves quickly away from zero as $m$
increases. This is because the cardinality of the set of the possible
confidence intervals increases and the approximation of the
necessarily discrete distribution of the reduced-bias estimator by a
Normal distribution with variance the inverse of the Fisher
information gets more accurate. This results in the increasing
accuracy of the approximation of the distribution of the Wald-pivot by
a Normal distribution.

As the current study demonstrates the Wald-type confidence intervals
based on any of the estimators do not behave satisfactorily for the
whole range of $\beta$ and for small sample sizes. For this
reason current research focuses on alternative confidence intervals
that can have one infinite endpoint (see
Section~\ref{discussionSec}). Until conclusive results are produced,
Wald-type confidence intervals based on the reduced-bias estimator can
still be used in practice as asymptotically correct, bearing in mind
that, they will be generally slightly conservative for moderate
effects (like the ones based on the maximum likelihood estimator)
especially in small samples, and also that their coverage properties
will deteriorate for extremely large effects.

\section{Shrinkage towards a binomial model for the end-categories}
Table~\ref{shrinkageTable} shows the maximum likelihood estimates, the
reduced-bias estimates and the corresponding estimated standard errors
from fitting a proportional odds model and a proportional hazards
model of the form~(\ref{2xk_model}) to the artificial data considered
in Example~\ref{artificial}.

\begin{table}[t!]
  \caption{Parameter estimates and corresponding estimated standard
    errors (in parenthesis) from fitting a proportional odds model and
    a proportional hazards model of the form~(\ref{2xk_model}) to the
    artificial data considered in Example~\ref{artificial}, using
    maximum likelihood and bias reduction.}
\begin{small}
\begin{center}
\begin{tabular}{cccc}
  \midrule\midrule
  Model & Parameter & Maximum likelihood & Bias reduction \\ \midrule
  \multirow{4}{*}{\parbox{4.5cm}{Proportional odds \\ ($G(\eta) = \exp(\eta)/\{1 +
      \exp(\eta)\}$)}} & $\beta$ & -1.944  (0.895) & -1.761 (0.850) \\
  & $\alpha_1$ & 1.187 (0.449) & 1.084 (0.428)  \\
  & $\alpha_2$ & 3.096 (0.787) & 2.781 (0.701)  \\
  & $\alpha_3$ & $\infty$ ($\infty$) & 4.457 (1.440) \\ \midrule
  \multirow{4}{*}{\parbox{4.5cm}{Proportional hazards \\ ($G(\eta) = 1
      - \exp\{-\exp(\eta)\}$)}} & $\beta$ & -0.689 (0.401) & -0.635 (0.389) \\
  & $\alpha_1$ & 0.313 (0.220) & 0.297 (0.219) \\
  & $\alpha_2$ & 1.097 (0.260) & 1.013 (0.246) \\
  & $\alpha_3$ & $\infty$ ($\infty$) & 1.518 (0.357) \\ \midrule\midrule
\end{tabular}
\end{center}
\end{small}
\label{shrinkageTable}
\end{table}

There is apparent shrinkage of the reduced-bias estimates towards
zero, which implies a shrinkage of the cumulative probabilities
towards $G(0)$. This implies a shrinkage of the probabilities for the
first and the last category of the ordinal scale towards $G(0)$ and
$1-G(0)$ respectively, and a corresponding shrinkage of the
probabilities of the intermediate categories towards zero.

To investigate further the apparent shrinkage effect, the maximum
likelihood and reduced-bias estimates of proportional odds and
proportional hazards models of the form~(\ref{2xk_model}) are obtained
for every possible $2 \times 6$ table with row totals $m_1 = m_2 =
3$. This setting is chosen because it is one that results in sparse
tables, allowing the construction of plots of fitted probabilities
that are not massively overcrowded (under this setting there are
$3136$ tables to be estimated).

For each category of the ordinal response, Figure~\ref{shrinkagePlots}
shows the fitted probabilities based on the reduced-bias estimator
against the fitted probabilities based on the maximum likelihood
estimator. The grey areas are where the points would all be expected
to lie if the shrinkage relationships were strictly satisfied for each
pair of fitted probabilities. Clearly this is not the case.

The points on the plots for the first category roughly lie slightly
above the $45^o$ line for fitted values less than $G(0)$, and slightly
below it for fitted values greater than $G(0)$.  The points for the
last category exhibit similar behaviour but with $G(0)$ replaced by
$1-G(0)$.  The shrinkage effect appears to be stronger the further the
probability is from the shrinkage points $G(0)$ and $1 - G(0)$.

The points on the plots for the intermediate categories lie mostly
under the $45^o$ line, except in cases where the maximum likelihood
fitted probability is very close to zero. Hence, the fitted
probabilities for the intermediate categories based on the
reduced-bias estimator tend to shrink towards zero.  The plots also
suggest that the further the probability is from zero the stronger is
the shrinkage effect.

\begin{figure}[t!]
  \caption{The fitted probabilities based on the reduced-bias
    estimator ($\tilde\pi_s$) against the fitted probabilities based
    on the maximum likelihood estimator ($\hat\pi_s$), for each
    category of the response. The top row corresponds to the
    proportional odds model and the bottom to a proportional hazards
    model. The grey areas are where the points would be expected to
    lie if shrinkage was strict.}
  \includegraphics[width = \textwidth]{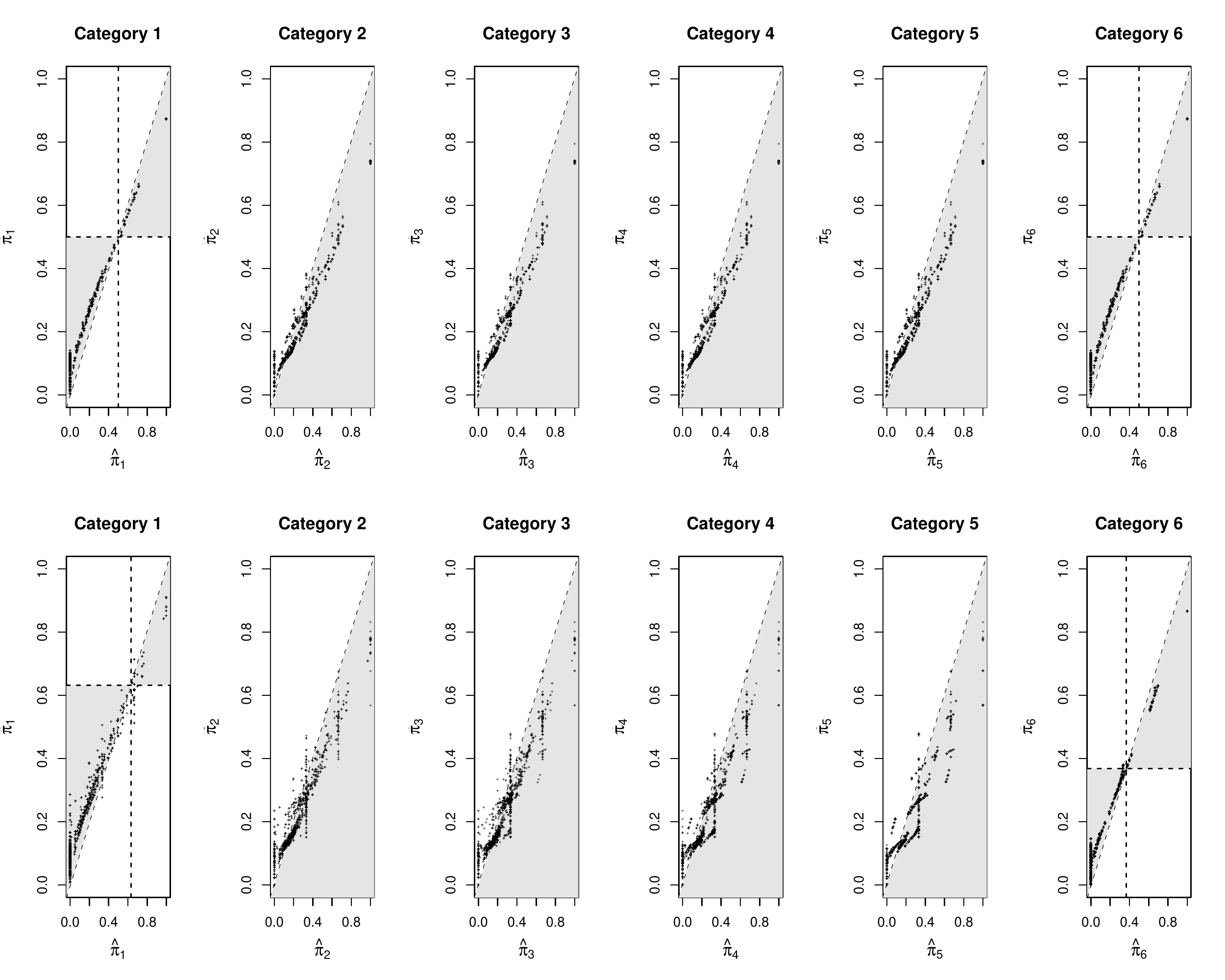}
  \label{shrinkagePlots}
\end{figure}

The shrinkage properties observed here are a direct generalization of
the shrinkage that is implied by improving bias in the estimation of
binomial logistic regression models \citep{copas:88, cordeiro:91,
  firth:92a} to links other than the logistic and to models with
ordinal responses.

Corresponding empirical investigations of shrinkage based on both
complete enumerations and simulations under models fitted to real data
have also been performed but are not shown here. The results are
qualitatively the same: reduction of bias in cumulative link models
shrinks the multinomial model towards a binomial model that has
probability $G(0)$ for the first category and probability $1 - G(0)$
for the last category.

\section{A simulation study}
\label{simulationSec}
In order to further illustrate the properties of the reduced-bias
estimator in more complex scenarios than the one in the complete
enumeration study of Section~\ref{completeenumeration}, a simulation
study is set-up based on part of the data that has been analyzed in
\citet{jackman:04}. The data is publicly available through the R
package \texttt{pscl} \citep{jackman:12} and seems to agree with the
data available for rater F1 in the analysis in \citet{jackman:04}. The
data contains the score of rater F1 for 106 applications to the
Political Science PhD Program at Stanford University along with
corresponding applicant-specific observations. The rater's score is on
a five-point integer-valued ordinal scale from $1$ to $5$, with $1$
indicating the lowest rating and $5$ indicating the highest
rating. Consider that the cumulative log-odds for rating $s$ for the
$r$th candidate is modelled as
\begin{equation}
  \label{ratings}
  \log\frac{\gamma_{rs}}{1 - \gamma_{rs}} = \alpha_s - \beta_1 x_{r1}
  - \beta_2 x_{r2} - \beta_3 z_{r1} - \beta_4 z_{r2} - \beta_5 g_{r}
  \quad (r = 1, \ldots, 106; s = 1, \ldots, 4)\, ,
\end{equation}
where $x_{r1}$ and $x_{r2}$ are the $r$th applicant's scores on the
quantitative and verbal section of the GRE, respectively (after
subtracting the respective mean and dividing by the respective
standard deviation), $z_{r1}$ and $z_{r2}$ are dummy variables
indicating whether the $r$th applicant has an interest in American
politics and Political Theory, respectively (with $1$ representing a
positive reply and $0$ a negative one), and $g_r$ is the gender of the
$r$th applicant $(r = 1, \ldots, 106)$. The parameter $\alpha_1,
\ldots, \alpha_5$ are the cutpoints and $\beta_1, \ldots, \beta_5$
describe the effect of the corresponding applicant-specific covariates
on the cumulative log-odds.
\begin{table}
  \caption{Estimated biases, mean-squared errors (MSE) and coverage
    probabilities of $95\%$ Wald-type confidence intervals from a
    simulation of size $10^5$ under the maximum likelihood fit of
    model~(\ref{ratings}). The last column shows the estimated relative
    increase of the mean squared error from its absolute minimum (the
    variance) due to bias. The relative increase of the mean squared
    error is the square of
    the bias divided by the variance. The estimated simulation error
    is less than 0.004 for the bias and the MSE estimates and less
    than 0.001 for the coverage estimates.}
  \begin{small}
    \begin{center}
      \begin{tabular}{ccrrcc}
        \midrule \midrule
        Method & Parameter & \multicolumn{1}{c}{Bias} &
        \multicolumn{1}{c}{MSE} &
        \multicolumn{1}{c}{Coverage} &
        \multicolumn{1}{c}{Bias$^2$/Variance (in \%)} \\ \midrule
        \multirow{5}{*}{Maximum likelihood} & $\beta_1$ & 0.132 & 0.142 &
        0.943 & 13.928 \\
        & $\beta_2$ & 0.055 & 0.062 & 0.943 & 5.203 \\
        & $\beta_3$ & 0.208 & 0.722 & 0.947 & 6.347 \\
        & $\beta_4$ & 0.004 & 0.630 & 0.944 & 0.003 \\
        & $\beta_5$ & 0.077 & 0.238 & 0.944 & 2.569 \\ \midrule
        \multirow{5}{*}{Bias correction} & $\beta_1$ & -0.001 & 0.106
        & 0.948 & 0.002 \\
        & $\beta_2$ & 0.001 & 0.051 & 0.953 & 0.001 \\
        & $\beta_3$ & -0.004 & 0.577 & 0.954 & 0.002 \\
        & $\beta_4$ & 0.003 & 0.551 & 0.956 & 0.001 \\
        & $\beta_5$ & 0.001 & 0.205 & 0.954 & 0.000 \\  \midrule
        \multirow{5}{*}{Bias reduction} & $\beta_1$ & 0.002 & 0.107 &
        0.949 &  0.002 \\
        & $\beta_2$ & 0.002 & 0.051 & 0.953 &  0.007 \\
        & $\beta_3$ & 0.002 & 0.579 & 0.954 &  0.001 \\
        & $\beta_4$ & 0.004 & 0.553 & 0.956 &  0.003\\
        & $\beta_5$ & 0.003 & 0.205 & 0.954 &  0.003 \\  \midrule \midrule
      \end{tabular}
    \end{center}
  \end{small}
  \label{simulation}
\end{table}

Model~(\ref{ratings}) is fitted using maximum likelihood and the
maximum likelihood estimates of $\beta_1, \ldots, \beta_5$ are
$1.993$, $0.892$, $2.816$, $0.009$, $1.215$ respectively indicating
that an increase in the value of any of the covariates is associated
with higher probability for high ratings holding all else in the model
fixed. Then an extensive simulation under the maximum likelihood fit
is performed for estimating the biases, mean squared errors and
coverage probabilities of Wald-type $95\%$ confidence intervals for
$\beta_1, \ldots, \beta_5$ when maximum likelihood, bias correction
and bias reduction is used. There have been instances of simulated
data sets where one or more rating categories were empty. In those
cases, empty categories were merged with neighbouring ones according
to the discussion in Remark 1 of
Section~\ref{completeenumeration}. The results are shown on
Table~\ref{simulation}. There was only one data set for which the
maximum likelihood estimate of $\beta_3$ was $+\infty$. This data set
was excluded when estimating the bias, mean squared error and coverage
probability for the maximum likelihood and the bias-corrected
estimator and hence the corresponding figures in the table estimate
the conditional respective quantities (that is given that the maximum
likelihood estimator has finite value). On the other hand, the
reduced-bias estimates were finite for all datasets and hence the
corresponding figures are estimates of the targeted unconditional
quantities. In this particular setting, the probability of the
conditioning event is rather small and a direct comparison of the
estimated conditional and unconditional quantities can be informative.

Temporarily ignoring the fact that the maximum likelihood estimator
can be infinite, both the bias corrected and reduced bias estimators
perform equally well in the current study. Furthermore, the figures in
Table~\ref{simulation} demonstrate a significant reduction both in
terms of bias and mean squared error when either bias correction or
bias reduction is used. In the current study the effect of estimation
bias is quite significant; the mean squared errors of the components
of the maximum likelihood estimator are inflated by as much as
$13.9\%$ due to bias from their minimum values (the variances). The
corresponding inflation factors for the bias corrected and
reduced-bias estimators are quite close to zero, which when combined
with the observed reduction in mean squared error illustrates the
benefits that the reduction of bias can have in the estimation of such
models. Lastly, a slight improvement in the coverage properties of
Wald-type confidence intervals is observed when the bias is
corrected.

Overall and taking into account that the reduced-bias estimator is
always finite the current study illustrates its superior frequentist
properties from the alternatives.

\section{Wine tasting data}
\label{wineData2}

\begin{table}[t]
  \caption{The reduced-bias estimates for the parameters of
    model~(\ref{partialProp}), the corresponding estimated standard
    errors (in parenthesis) and the values of the $Z$ statistic
    for the hypothesis that the corresponding parameter is
    zero.}
  \begin{small}
  \begin{center}
    \begin{tabular}{crrr}
\midrule\midrule
Parameter & \multicolumn{2}{c}{RB estimates} & $Z$ statistic \\ \midrule
$\alpha_1$ & -1.19 & (0.50) & -2.40 \\
$\alpha_2$ &  1.06 & (0.44) & 2.42 \\
$\alpha_3$ &  3.50 & (0.74) & 4.73 \\
$\alpha_4$ &  5.20 & (1.47) & 3.52 \\
$\beta_1$ &  2.62 & (1.52) & 1.72 \\
$\beta_2$ &  2.05 & (0.58) & 3.54 \\
$\beta_3$ &  2.65 & (0.75) & 3.51 \\
$\beta_4$ &  2.96 & (1.50) & 1.98 \\
$\theta$ & 1.40 & (0.46) & 3.02 \\
\midrule\midrule
\end{tabular}
\end{center}
\end{small}
\label{wineDataRB}
\end{table}

The partial proportional odds model of Example~\ref{wineExample} is
here refitted using the reduced-bias estimator. The result is shown
in Table~\ref{wineDataRB}. All estimates and estimated standard errors
are finite. A Wald statistic for testing departures from the
assumption of proportional odds via departures from the hypothesis
$\beta_1 = \beta_2 = \beta_3 = \beta_4$ is
\[
W = (L\tilde{\bb\beta}_{RB})^T I(\tilde{\bb\delta}_{RB})
L\tilde{\bb\beta}_{RB} \, ,
\]
where
\[
L = \left[
\begin{array}{cccc}
1 & 0 & 0 & -1 \\
0 & 1 & 0 & -1 \\
0 & 0 & 1 & -1
\end{array}
\right]
\]
is a matrix of contrasts of $\beta$. The matrix
\[
I({\bb\delta}) = \left\{LF^{{\bb\beta}{\bb\beta}}({\bb\delta})L^T\right\}^{-1}\, ,
\]
is the inverse of the variance-covariance matrix of the asymptotic
distribution of $L\tilde{\bb\beta}_{RB}$, where
$F^{{\bb\beta}{\bb\beta}}({\bb\delta})$ is the ${\bb\beta}$-block of
the inverse of the Fisher information. By the asymptotic normality of
$\tilde{\bb\beta}_{RB}$, $W$ has an asymptotic $\chi^2$ distribution
with $3$ degrees of freedom. The value of $W$ for the data in
Table~\ref{wineData} is $0.7502$ leading to a $p$-value of $0.861$,
which provides no evidence against the proportional odds
assumption. This is also apparent from Table~\ref{wineDataRB} where
the reduced-bias estimates of $\beta_1, \beta_2, \beta_3, \beta_4$ are
comparable in value.

It is worth noting that, in contrast to the output reported in
Example~\ref{wineExample}, the values of the $Z$ statistics for
$\alpha_4$, $\beta_1$ and $\beta_4$ are far from being exactly zero.

\section{Concluding remarks and further work}
\label{discussionSec}
Based on the results of the complete enumeration study,
$\tilde\beta_{RB}$ appears to be always finite in contrast to
$\hat\beta_{ML}$ and $\tilde\beta_{BC}$, and also to have comparable
behaviour to $\hat\beta_{EL}$ in terms of bias and mean squared
error. Furthermore, Wald-type asymptotic confidence intervals based on
$\tilde\beta_{RB}$ behave satisfactorily, maintaining good coverage
properties for a wide range of $\beta$ values. A complete loss of
coverage is still present but the point where this happens is far away
from zero and diverges as the number of observations increases. The
application of the current complete enumeration setup for
complementary log-log and probit link functions, for varying values of
the row totals, and for different numbers of categories, resulted in
qualitatively the same conclusions.

In Remark 1 of the complete enumeration study, the finiteness of the
reduced-bias estimates for $\alpha_1$ and/or $\alpha_q$ was noted even
in cases where the first and/or last category of the ordinal variable
is not observed. This behaviour has been also encountered in the
simulation study of Section~\ref{simulationSec} and in all of the many
settings where the reduced-bias estimator has been applied, and is
defensible from an experimental point of view. When the experimenter
sets an ordinal scale, the end-categories of that scale largely
determine the possible responses. Hence, one might argue that the end
categories should play a bigger role than the intermediate categories
in the analysis, and a good estimation method should not be as
democratic as maximum likelihood is in this respect; accepting that
the ordinal scale is well-defined, if an end category is not observed
then it seems more appropriate to slightly inflate its probability of
occurrence, instead of setting it to zero as the maximum likelihood
estimator would do.

The latter point of view does not only apply for non-observed
end-categories. It applies to all analyses of ordinal data through
cumulative link models and is reinforced by the fact that an
improvement in the frequentist properties of the maximum likelihood
estimator resulted in the shrinkage of the cumulative link model
towards a binomial model for the end-categories.

The above observations, along with the fact that $\tilde{\bb\delta}_{RB}$
respects the invariance properties of the cumulative link model and
can be easily obtained using the procedures in
Section~\ref{implementation}, provide a strong case for its routine
use in the estimation of cumulative link models.

\citet{laara:85} demonstrate the equivalence of continuation ratio
models with complementary log-log link and proportional hazards models
in discrete time. Hence, the reduced-bias estimates for the regression
parameters of the former can be obtained by using the results in the
current paper for the latter.

The investigation of confidence intervals that maintain good
properties without suffering from a complete loss of coverage for
extreme effects is the subject of future work. Current research
focuses on the use of profiles of the asymptotic pivot
$\{U^*({\bb\delta})\}^T F^{-1}({\bb\delta}) U^*({\bb\delta})$ which
can be shown to have an asymptotic $\chi^2$ distribution, and the
combination of the resultant intervals with the profile likelihood
intervals. In this way confidence intervals with one infinite endpoint
are possible and are suggested to accompany the reduced-bias estimator
which appears always to take finite value. Such intervals seem to
better reflect uncertainty when extreme settings are considered, and
lead to improved coverage properties without loss of coverage.

As is done in Subsection~\ref{wineData2}, comparison of nested models
can be performed using asymptotic Wald-type test based on the
reduced-bias estimator. Another option is the use of the adjusted
score statistic
\[
\{U^*(\tilde{\bb\delta}_{-})\}^T
F^{-1}(\tilde{\bb\delta}_{-}) U^*(\tilde{\bb\delta}_{-})\,,
\]
where $\tilde{\bb\delta}_{-}$ are the estimates under the hypothesis
that results in the smaller model, and $U^*({\bb\delta})$ and
$F({\bb\delta})$ are the vector of adjusted score functions and the
Fisher information of the larger model, respectively. The fact that
$U^*({\bb\delta}) = U({\bb\delta}) + A({\bb\delta})$ where
$A({\bb\delta}) = O(1)$ guarantees an asymptotic $\chi^2$ distribution
for that statistic. For the example in Subsection~\ref{wineData2} the
value of the adjusted score statistic is $0.9357$ on $3$ degrees of
freedom giving a $p$-value of 0.8168 which leads to qualitatively the
same conclusion as that of the Wald test. When testing departures from
the proportional odds assumption in general, the adjusted score
statistic has the same disadvantage as the ordinary score statistic;
the Fisher information matrix for the partial proportional odds model
can be non-invertible when evaluated at the estimates of the
corresponding proportional odds model.

\section*{Acknowledgments}
The author is grateful to two anonymous referees and the Associate
Editor whose comments have significantly improved the current
work. Furthermore, the author is grateful to David Firth for the
helpful and stimulating discussions on this work, and to Cristiano
Varin and Thomas W. Yee for their comments on this paper.

Part of this work was completed between September 2007 and September
2010 when the author was a CRiSM Research Fellow at University of
Warwick, and between January 2012 and July 2012 when the author was
Senior Research Fellow at the Department of Statistics of the
University of Warwick. The support of EPSRC is gratefully acknowledged
for funding both positions.

\section*{Supplementary material}
The accompanying supplementary material includes an R script
\citep{rproject:12} that can be used to produce the results of the
complete enumeration study for any number of categories, any link
function and any configuration of row totals in contingency
tables. The current version of the R function bpolr is also
included. The bpolr function fits cumulative link models and their
extensions with dispersion effects either by maximum likelihood, or
bias reduction or bias correction. An updated version of the function
will be part of the next major release of the R package brglm
\citep{brglm:07}. Scripts that reproduce the data analyses undertaken
in the paper are also available in the supplementary material.

\bibliographystyle{chicago}
\begin{small}
\bibliography{references.bib}
\end{small}

\end{document}